\documentclass[11pt]{article}
\pdfoutput=1
\usepackage[T1]{fontenc}
\usepackage{lmodern}
\usepackage{comment}
\usepackage{slantsc}
\usepackage[protrusion=true,expansion=true]{microtype}
\usepackage{breakcites}
\usepackage{amsmath,amssymb,amsfonts,amsthm, mathtools}
\usepackage{subcaption}
\usepackage{graphicx}
\usepackage{fullpage}
\usepackage{setspace}
\usepackage[backref=page]{hyperref}
\usepackage[nameinlink]{cleveref}
\usepackage{color}
\usepackage{wrapfig}
\usepackage{tikz}
\usetikzlibrary{decorations.pathreplacing,patterns}
\usepackage{algorithm}
\usepackage[noend]{algpseudocode}
\usepackage[framemethod=tikz]{mdframed}
\usepackage{xspace}
\usepackage{pgfplots}
\usepackage{framed}
\usepackage{thmtools}
\usepackage{thm-restate}
\usepackage{tabu}
\usepackage{fancyhdr}
\usepackage{placeins}
\usepackage{float}
\pgfplotsset{compat=1.5}

\newtheorem{theorem}{Theorem}[section]

\newtheorem{lemma}[theorem]{Lemma}

\newtheorem{definition}[theorem]{Definition}
\newtheorem{remark}[theorem]{Remark}

\newtheorem{invariant}[theorem]{Invariant}

\newtheorem{conjecture}[theorem]{Conjecture}

\newenvironment{proofof}[1]{\begin{trivlist} \item {\bf Proof
#1:~~}}
  {\qed\end{trivlist}}

\newcommand{\namedref}[2]{\hyperref[#2]{#1~\ref*{#2}}}

\newcommand{\lemlab}[1]{\label{lem:#1}}
\newcommand{\lemref}[1]{\namedref{Lemma}{lem:#1}}

\newcommand{\deflab}[1]{\label{def:#1}}

\def \EstLevel    {\mdef{\mathsf{EstLevel}}}

\def \AMS    {\mdef{\mathsf{AMS}}}

\def \MaintainIter    {\mdef{\mathsf{MaintainIter}}}

\newcommand\norm[1]{\left\lVert#1\right\rVert}
\newcommand{\PPr}[1]{\ensuremath{\mathbf{Pr}\left[#1\right]}}

\renewcommand{\O}[1]{\ensuremath{\mathcal{O}\left(#1\right)}}
\newcommand{\tO}[1]{\ensuremath{\tilde{\mathcal{O}}\left(#1\right)}}
\newcommand{\eps}{\varepsilon}

\def \calA    {\mdef{\mathcal{A}}}

\def \calF    {\mdef{\mathcal{F}}}

\def \calP    {\mdef{\mathcal{P}}}

\def \bA    {\mdef{\mathbf{A}}}
\def \bB    {\mdef{\mathbf{B}}}

\def \be    {\mdef{\mathbf{e}}}

\def \bq    {\mdef{\mathbf{q}}}

\def \bu    {\mdef{\mathbf{u}}}

\def \bv    {\mdef{\mathbf{v}}}
\def \bx    {\mdef{\mathbf{x}}}
\def \by    {\mdef{\mathbf{y}}}
\def \bz    {\mdef{\mathbf{z}}}

\newcommand{\mdef}[1]{{\ensuremath{#1}}\xspace}  

\DeclareMathOperator*{\polylog}{polylog}
\DeclareMathOperator*{\poly}{poly}

\newcommand{\ignore}[1]{}

\newif\ifnotes\notestrue 
\ifnotes
\newcommand{\samson}[1]{\textcolor{red}{{\bf (Samson:} {#1}{\bf ) }} \marginpar{\tiny\bf
             \begin{minipage}[t]{0.5in}
               \raggedright S:
            \end{minipage}}}

\newcommand{\david}[1]{\textcolor{purple}{{\bf (David:} {#1}{\bf ) }} \marginpar{\tiny\bf
             \begin{minipage}[t]{0.5in}
               \raggedright D:
            \end{minipage}}}

\else
\newcommand{\samson}[1]{}
\newcommand{\david}[1]{}
\fi

\makeatletter
\renewcommand*{\@fnsymbol}[1]{\textcolor{mahogany}{\ensuremath{\ifcase#1\or *\or \dagger\or \ddagger\or
 \mathsection\or \triangledown\or \mathparagraph\or \|\or **\or \dagger\dagger
   \or \ddagger\ddagger \else\@ctrerr\fi}}}
\makeatother

\providecommand{\email}[1]{\href{mailto:#1}{\nolinkurl{#1}\xspace}}

\definecolor{mahogany}{rgb}{0.75, 0.25, 0.0}
\definecolor{darkblue}{rgb}{0.0, 0.0, 0.55}
\definecolor{darkpastelgreen}{rgb}{0.01, 0.75, 0.24}
\definecolor{darkgreen}{rgb}{0.0, 0.2, 0.13}
\definecolor{darkgoldenrod}{rgb}{0.72, 0.53, 0.04}
\definecolor{darkred}{rgb}{0.55, 0.0, 0.0}
\definecolor{forestgreenweb}{rgb}{0.13, 0.55, 0.13}
\definecolor{greencss}{rgb}{0.0, 0.5, 0.0}
\definecolor{bleudefrance}{rgb}{0.19, 0.55, 0.91}
\definecolor{darkpastelpurple}{rgb}{0.59, 0.44, 0.84}
\definecolor{darkcerulean}{rgb}{0.03, 0.27, 0.49}

\hypersetup{
     colorlinks   = true,
     citecolor    = mahogany,
	 linkcolor	  = forestgreenweb
}

\fancypagestyle{pg}
{
\lhead{}
\rhead{}
\cfoot{--\ \thepage\ --}

}

\AtBeginDocument{%
  \DeclareFontShape{T1}{lmr}{m}{scit}{<->ssub*lmr/m/scsl}{}%
}
\begin{document}

\allowdisplaybreaks

\title{Adversarial Robustness on Insertion-Deletion Streams}
\author{Elena Gribelyuk \\ Princeton University \\ \email{eg5539@princeton.edu}  
\and 
Honghao Lin \\ Carnegie Mellon University \\ \email{honghaol@andrew.cmu.edu} 
\and 
David P. Woodruff \\ Carnegie Mellon University \\ \email{dwoodruf@andrew.cmu.edu}
\and
Huacheng Yu \\ Princeton University \\ \email{hy2@cs.princeton.edu}
\and
Samson Zhou \\ Texas A\&M University \\ \email{samsonzhou@gmail.com}}
\date{\today}

\maketitle
\pagestyle{pg}

\begin{abstract}
We study adversarially robust algorithms for insertion-deletion (turnstile) streams, where future updates may depend on past algorithm outputs. While robust algorithms exist for insertion-only streams with only a polylogarithmic overhead in memory over non-robust algorithms, it was widely conjectured that turnstile streams of length polynomial in the universe size $n$ require space linear in $n$. We refute this conjecture, showing that robustness can be achieved using space which is significantly sublinear in $n$. Our framework combines multiple linear sketches in a novel estimator-corrector-learner framework, yielding the first insertion-deletion algorithms that approximate: (1) the second moment $F_2$ up to a $(1+\varepsilon)$-factor in polylogarithmic space, (2) any symmetric function $\cal{F}$ with an $\mathcal{O}(1)$-approximate triangle inequality up to a $2^{\mathcal{O}(C)}$ factor in $\tilde{\mathcal{O}}(n^{1/C}) \cdot S(n)$ bits of space, where $S$ is the space required to approximate $\cal{F}$ non-robustly; this includes a broad class of functions such as the $L_1$-norm, the support size $F_0$, and non-normed losses such as the $M$-estimators, and (3) $L_2$ heavy hitters. For the $F_2$ moment, our algorithm is optimal up to $\textrm{poly}((\log n)/\varepsilon)$ factors. Given the recent results of Gribelyuk et al. (STOC, 2025), this shows an exponential separation between linear sketches and non-linear sketches for achieving adversarial robustness in turnstile streams.  
\end{abstract}


\section{Introduction}
In the classical streaming model, updates arrive sequentially and define an implicit dataset, and the goal is to approximate or compute a specific statistic using space sublinear in both the stream length $m$ and the dataset dimension $n$.
Typically, the algorithm is expected to process the data in a single pass, maintaining a compact summary that enables efficient estimation despite stringent memory constraints. 
This model effectively captures the computational limitations encountered in many large-scale applications, where the data size far exceeds available storage. 
As a result, the streaming model has attracted significant attention in both theoretical and applied research, leading to a rich body of work exploring algorithmic techniques, lower bounds, and practical implementations of streaming algorithms. 
Such applications include processing logs from network and traffic monitoring, analyzing financial market transactions, handling sensor data from Internet of Things (IoT) devices, and managing scientific measurement streams.

While classical streaming algorithms typically assume that the input stream is fixed and independent of the internal parameters of the algorithm, many practical settings require a more flexible model, where previous algorithmic outputs may influence future inputs. 
For example, database queries often depend on prior responses, iterative methods like stochastic gradient descent update their internal state based on earlier results, recommendation systems create feedback loops through user interactions, and financial markets respond to aggregated statistics, which in turn shape future input data. 
This juxtaposition between the assumption of an oblivious, fixed input dataset and the adaptive nature of many real-world applications undermines the foundations of classical streaming analysis, which often crucially leverages the independence between the internal randomness of the algorithm and the underlying dataset. 

\paragraph{Adversarial robustness.}
Motivated by this gap, \cite{Ben-EliezerJWY22} introduced the \emph{adversarially robust streaming model} to capture settings in which the sequence of updates may be adaptive or even adversarial. 
This model and its variants has received significant attention, with a growing body of work developing robust algorithms for various tasks~\cite{MironovNS11,MitrovicBNTC17,AvdiukhinMYZ19,Ben-EliezerY20,CherapanamjeriN20,AlonBDMNY21,BeimelKMNSS22,ChakrabartiGS22,MenuhinN22,NaorO22,AssadiCGS23,CherapanamjeriSWZZZ23,CohenNSS23,DinurSWZ23,PengR23,WoodruffZZ23b,AhmadianCohen2024,CohenNSSS24,CohenSS25,Ben-EliezerSO25}. 
Formally, at each time step $t\in[m]$, the streaming algorithm $\calA$ receives an update $u_t = (a_t, \Delta_t)$, where $a_t \in [n]$ is an index and $\Delta_t \in \mathbb{Z}$ denotes an increment or decrement to coordinate $a_t$ of an implicit frequency vector $\bx$, defined by $x_i = \sum_{t: a_t = i} \Delta_t$. 
In this manner, we also define the prefix vector $\bx^{(t)}$ as the frequency vector after the first $t$ updates, i.e., $x^{(t)}_i = \sum_{s \le t: a_s = i} \Delta_s$. 
Throughout, we assume that $m = \poly(n)$ and that $|\Delta_t| \le \poly(n)$ for all $t \in [m]$; by scaling, one could equivalently assume that each $\Delta_t$ is a multiple of $\frac{1}{\poly(n)}$.
We say the algorithm $\calA$ is \emph{adversarially robust} for a specified function $f:\mathbb{Z}^n \to \mathbb{R}$ if $\calA$ provides accurate estimates of $f(\bx^{(t)})$ at each time step $t\in[m]$, even when the updates are chosen adaptively and may depend on the entire transcript of previous updates and algorithm responses.

\begin{definition}
\cite{Ben-EliezerJWY22}
Let $f: \mathbb{Z}^n \to\mathbb{R}$ be a fixed function. 
A streaming algorithm $\calA$ is adversarially robust if, for any $\eps,\delta\in(0,1)$ and $m=\poly(n)$, at each time step $t \in [m]$, the algorithm outputs an estimate $Z_t$ such that
\[\PPr{|Z_t - f(\bx^{(t)})|\le\eps\cdot f(\bx^{(t)})}\ge1-\delta.\]
\end{definition}
Conceptually, the model can be viewed as a sequential two-player game between the randomized streaming algorithm $\calA$ and an adaptive adversary, which constructs a sequence of updates $\{u_1, \ldots, u_m\}$, where each $u_t$ may depend on the entire history of past updates and algorithm outputs. 
Then, the game unfolds as follows: 
\begin{enumerate}
\item In each round $t$, the adversary selects an update $u_t$ based on the previous updates $u_1, \ldots, u_{t-1}$ and the corresponding outputs $Z_1, \ldots, Z_{t-1}$ of the algorithm.
\item 
The algorithm $\calA$ processes $u_t$ and updates its internal state.
\item
The algorithm $\calA$ then returns an estimate $Z_t$ of $f(\bx^{(t)})$ before proceeding to the next round.
\end{enumerate}
By nature of the game, the algorithm is only permitted a single pass over the data stream. 

\paragraph{Insertion-only streams.}
In the insertion-only streaming model, each update $\Delta_t>0$ increases some coordinate of the underlying frequency vector. 
This setting has received significant attention recently, and it is now known that many fundamental streaming problems admit adversarially robust algorithms using sublinear space \cite{HassidimKMMS20,BravermanHMSSZ21,KaplanMNS21,WoodruffZ21,AjtaiBJSSWZ22,Ben-EliezerJWY22,AttiasCSS23,JiangPW23}. 
The techniques range from importance sampling and merge-and-reduce~\cite{BravermanHMSSZ21,JiangPW23}, to switching between various sketches or bounding the total number of computation paths~\cite{Ben-EliezerJWY22}, to differential privacy~\cite{HassidimKMMS20,AttiasCSS23}, to new algorithmic paradigms for estimating the difference of the function values at different times~\cite{WoodruffZ21}. 
Notably, for many problems such as heavy hitters, entropy estimation, clustering, linear regression, subspace embeddings, and graph sparsification, there exist adversarially robust solutions, which compared to the space complexity of classical insertion-only algorithms, only incur a small overhead polylogarithmic in the universe size $n$. 

A particularly well-studied family of problems is that of estimating the frequency moments $F_p(\bx)=\sum_{i\in[n]}|x_i|^p$ for $p>0$ and $F_0(\bx)=|\{i\in[n]\,\mid\,x_i\neq 0\}|$. 
These moments capture various important statistics of the data stream. 
For example, $F_1$ corresponds to the total number of inserted items, $F_2$ relates to the repeat rate or collision probability, and $F_0$ is the number of distinct elements, i.e., the support size of the stream. 
Moreover, the $L_p$ norms satisfy $\|\bx\|_p=(F_p(\bx))^{1/p}$, so that algorithms for estimating the $F_p$ moments directly yield algorithms for estimating $L_p$ norms, and vice versa.
Since the seminal work of \cite{AlonMS99} which formalized the classical streaming model and introduced the frequency moment estimation problem, this family of problems has been studied extensively~\cite{Indyk06,IndykW05,Li08,KaneNW10a,KaneNW10b,AndoniKO11,Ganguly11,Ganguly12,BravermanKSV14,BravermanVWY18,BravermanGLWZ18,GangulyW18,WoodruffZ21b,JayaramWZ24,BravermanZ25,FengSW25}. More recently, there has been significant work on designing robust algorithms for frequency moments~\cite{HassidimKMMS20,WoodruffZ21,Ben-EliezerEO22,Ben-EliezerJWY22,WoodruffZ24}, and indeed it has been shown~\cite{WoodruffZ21} that achieving a $(1+\eps)$-multiplicative approximation for any input accuracy parameter $\eps\in(0,1)$ incurs at most $\polylog\left(\frac{1}{\eps}\right)$ space overhead in an insertion-only stream.
This makes the $F_p$ moment estimation problem a compelling case study for understanding the limits of robustness in the streaming model. 

\paragraph{Insertion-deletion streams.}  
Compared to insertion-only streams, significantly less is understood about the space complexity of adversarially robust algorithms in turnstile streams, where updates $\Delta_t$ may increase or decrease coordinates.
\cite{Ben-EliezerJWY22} showed that when deletions can only decrease the underlying statistic by at most some bounded factor, there exists an adversarially robust algorithm with space sublinear in the stream length $m$; however, this assumption may not hold in general. 
By utilizing techniques in differential privacy, \cite{HassidimKMMS20,Ben-EliezerEO22,WoodruffZ24} designed adversarially robust algorithms with space sublinear in the stream length $m$; however, when $m$ grows polynomially in the dimension $n$, the resulting space complexity is no longer sublinear in $n$.

On the negative side, a recent line of work~\cite{HardtW13,CohenNSSS24,GribelyukLWYZ24,GribelyukLWYZ25} has developed adaptive attacks against linear sketches, which represent a popular algorithmic tool for designing low-space streaming algorithms. 
Given a stream of updates to a vector $\bx\in \mathbb{R}^n$, a linear sketch compresses the data into a much smaller vector $\bA\bx \in \mathbb{R}^r$, where $r \ll n$, using a carefully chosen matrix $\bA \in \mathbb{R}^{r \times n}$ that is efficient to store and apply. 
The matrix should be designed so that there exists a post-processing function $g$, so that $g(\bA\bx)$ can be used to estimate the desired quantity about $\bx$, e.g., $g(\bA\bx)\approx\|\bx\|_p$. 
Since $\bA\bx$ only has $r$ entries, the algorithm often only needs to store $\O{r\log n}$ bits throughout the stream, and for many interesting problems, it is possible to achieve $r\ll n$. 
In fact, all known algorithms for turnstile streams of sufficiently large length are based on linear sketches. 
Moreover, under certain conditions, it has been shown that optimal algorithms for turnstile streaming can be characterized entirely by linear sketching methods~\cite{LiNW14,AiHLW16,HosseiniLY19,KallaugherP20}.
Unfortunately, results by \cite{HardtW13,GribelyukLWYZ24,GribelyukLWYZ25} prove that any adversarially robust insertion-deletion streaming algorithm that uses a linear sketch to achieve any constant multiplicative approximation to the $F_p$ moments requires sketching dimension $r=\Omega(n)$. 
Although it is not clear whether \emph{adversarially robust} turnstile streaming algorithms can be completely characterized by linear sketching methods, a long-standing open conjecture~\cite{stoc2021workshop,focs2023workshop,encore2025workshop} has been:
\begin{mdframed}
\vspace{-0.12in}
\begin{conjecture}
\label{conj:main}
Any adversarially robust constant-factor approximation algorithm for $F_p$ moments on insertion-deletion streams of sufficient length $m=\poly(n)$ requires $\Omega(n)$ space.
\end{conjecture}
\end{mdframed}

\subsection{Our Contributions}
In this work, we show that \Cref{conj:main} is false. 
Namely, we circumvent the above $\Omega(n)$ lower bound for linear sketches and design sublinear-space streaming algorithms that output a $(1+\eps)$-approximations to a wide variety of functions in an adversarial insertion-deletion stream. 
Crucially, our algorithms are {\it non-linear}, which is necessary to bypass the linear sketch lower bounds for achieving adversarial robustness in insertion-deletion streams in \cite{GribelyukLWYZ25}.

We first present our result for arguably one of the most fundamental problems for streaming algorithms, estimating the $F_2$ moment $\|\bx\|_2^2=x_1^2+\ldots+x_n^2$. 
The $F_2$ moment estimation problem was first studied in the seminal work of \cite{AlonMS99}, which formally defined the streaming model and gave a $\O{\frac{1}{\eps^2} \log n \log \frac{1}{\delta}}$ space algorithm for $F_2$ estimation, which was later shown to be optimal \cite{KaneNW10a,BravermanZ25}. 

\begin{restatable}{theorem}{thmltwo}
\label{thm:ltwo}
Given any $\eps\in(0,1)$, there exists an adversarially robust insertion-deletion streaming algorithm on a stream of length $m$ that with high probability, outputs a $(1+\eps)$-approximation to the $F_2$ moment at all times, for the underlying frequency vector of universe size $n$. 
For $m=\poly(n)$, the algorithm uses $\poly\left(\frac{1}{\eps},\log n\right)$ bits of space. 
\end{restatable}
As the lower bound by \cite{BravermanZ25} implies that even in the non-adaptive setting, $F_2$ moment estimation requires $\poly\left(\frac{1}{\eps},\log n\right)$ bits of space, our results in \Cref{thm:ltwo} are optimal up to $\poly\left(\frac{1}{\eps},\log n\right)$ factors. 
Interestingly, while our overall framework is not a linear sketch, it combines multiple linear sketches in careful way to only require each of them to be robust to a small number of adaptive queries. 
Thus, while prior work~\cite{LiNW14,AiHLW16,KannanMSY18,KallaugherP20} showed that optimal turnstile streaming algorithms are necessarily linear sketches, our results imply that this characterization breaks down in the adversarially robust setting: robustness can be achieved through non-linear frameworks that fundamentally go beyond linear sketching.

We next show that our robust frequency moment algorithms can be used to robustly identify the heavy-hitters at each step in an adaptive insertion-deletion stream. 
Recall that for a frequency vector $\bx\in\mathbb{R}^n$ and a heavy-hitter threshold $\eps\in(0,1)$, the goal of the heavy-hitter problem is to recover a ``small'' list containing all coordinates $i\in[n]$ such that $|x_i|\ge\eps\cdot\|\bx\|_2$. 

\begin{restatable}{theorem}{thmhh}
\label{thm:hh}
Given any $\eps \in (0,1)$, there exists an adversarially robust insertion-deletion streaming algorithm on a stream of length $m$ that with high probability, outputs the set of $L_2$ heavy hitters at all times, for the underlying frequency vector of universe size $n$. 
For $m = \poly(n)$, the algorithm uses $\poly\left(\frac{1}{\eps}, \log n \right)$ bits of space.
\end{restatable}
Since an algorithm for $L_2$ heavy hitters can be used to recover $L_p$ heavy hitters for $p \leq 2$, we obtain an adversarially robust algorithm for $L_p$ heavy hitters for all $p \leq 2$. 
Again, we remark that the best-known $L_2$ heavy-hitter algorithms use $\poly\left(\frac{1}{\eps},\log n\right)$ space for $p = 2$ and such dependencies are known to be necessary~\cite{Ganguly12,WoodruffZ21b}. 
Thus, our result essentially matches the optimal space bounds of classical non-adaptive $L_2$ heavy-hitter algorithms, up to polynomial factors in $\frac{1}{\eps}$ and logarithmic factors in $n$. 

More generally, we show that our approach applies broadly to any function that satisfies an approximate triangle inequality and admits a non-adaptive sketching algorithm.
Formally, we say that for a constant $\beta>0$, a symmetric function $\calF$ satisfies a $\beta$-approximate triangle inequality if $\calF(\bx-\bz)\le\beta\cdot(\calF(\bx-\by)+\calF(\by-\bz))$ for all vectors $\bx,\by,\bz\in\mathbb{R}^n$. 
\begin{restatable}{theorem}{thmtri}
\label{thm:tri}
Suppose $\calF$ is a symmetric function on a stream of length $m=\poly(n)$, has value $\left[\frac{1}{\poly(m)},\poly(m)\right]$ and satisfies the $\beta$-approximate triangle inequality. 
Let $\kappa>2\beta+1$ be a fixed constant and suppose there exists a non-adaptive turnstile streaming algorithm that uses $S(n)\cdot\log\frac{1}{\delta}$ bits of space and with probability at least $1-\delta$, outputs a $\kappa$-approximation to $\calF$. 
Then for any constant $C>1$, there exists an adversarially robust insertion-deletion streaming algorithm on a stream of length $m$ that uses $\tO{n^{1/C}}\cdot S(n)$ bits of space and with high probability, outputs a $\kappa^{\O{c}}$-approximation to $\calF$ at all times.  
\end{restatable}

We apply \Cref{thm:tri} to the $L_1$ norm, the support size $L_0$, the $L_p$ norms for $p \in (0,1)$, and many other loss functions such as the Huber loss functions, M-estimators, and so on. 

\subsection{Technical Overview}
Recall that from the previous discussion, given a data stream that induces a frequency vector $\bx\in\mathbb{R}^n$, a common algorithmic approach is to generate a random sketch matrix $\bA\in\mathbb{R}^{r\times n}$ and maintain $\bA\bx$ throughout the stream. 
The matrix should be designed so that there exists a post-processing function $g$, so that $g(\bA\bx)$ can be used to estimate the desired quantity for $\bx$, e.g., $g(\bA\bx)\approx\|\bx\|_p$. 
Since $\bA\bx$ only has $r$ entries, the algorithm often only needs to store $\O{r\log n}$ bits throughout the stream, and for many interesting problems, it is possible to achieve $r\ll n$. 
For example, constant-factor approximations to the $L_p$ norm are possible with  $r=\polylog(n)$ for $p\le 2$ and $r=\O{n^{1-2/p}}$ for $p>2$. 
Unfortunately, recent results by \cite{GribelyukLWYZ24,GribelyukLWYZ25} prove that any such linear sketch for adversarially robust insertion-deletion streams requires dimension $\Omega(n)$. 
One inherent reason is that linear sketches only use randomness at the start to generate the sketching matrices, thereafter, the algorithm is deterministic.
In order for the adversary to break such algorithms, it suffices to take $\poly(n)$ queries to gradually figure out these initial random bits, and then construct a hard query for the corresponding sketching matrix.
This is how the recent lower bounds were proved.
Thus, to circumvent the lower bounds, it informally seems like new random bits need to be introduced during the execution of the streaming algorithm, rather than all at the beginning. 
Our solution achieves small space by initializing new independent sketches in the ``middle'' of the data stream. 

To that end, suppose that the underlying frequency vector $\bx$ satisfies $\bx=\bz+\bq$, where $\bz$ is the frequency vector defined by all of the updates of the stream, up to some time $t$ and $\bq$ is the frequency vector defined by the remaining updates after $t$, i.e., $\bz$ corresponds to some prefix of the stream and $\bq$ corresponds to the remaining suffix of the stream. 
Consider a sketch matrix $\bB$ that was initialized at time $t+1$, so that the algorithm tracks $\bB\bq$. 
For a random sketch matrix $\bB$ designed for $F_2$ moment estimation, if we knew $\bB\bz$ (which we actually do not, since $\bB$ is only initialized after $\bz$ appears), then since $\bB\bx=\bB\bz+\bB\bq$, we could approximate $\|\bx\|_2^2$ from $\bB\bx$.
Below, we will show that either we can approximate $\|\bx\|_2^2$ without knowing $\bB\bz$, or we will gradually learn $\bB\bz$. 

\subsubsection{Technical Overview for \texorpdfstring{$F_2$}{F2} Moment Estimation}
Now suppose there exists another random sketch matrix $\bA$ designed for $F_2$ moment estimation, but initialized at the beginning of the stream, so that the algorithm is able to store $\bA\bz$ as well as to track $\bA\bx$ after time $t$. 
In this case, the information we have about $\bz$ is the sketch $\bA\bz$.
Suppose $\bq$ has a random direction independent of $\bz$, then we cannot hope to learn anything about $\bB\bz$ from it.
However, this is exactly a case where we do not need to know $\bB\bz$ to approximate $\|\bx\|_2^2$.
Since random directions are nearly orthogonal in large dimensions, we have $\|\bz+\bq\|_2^2\approx \|\bz\|_2^2+\|\bq\|_2^2$, which in turn, can be estimated from $\bA\bz$ and $\bB\bq$ respectively.
On the other hand, if $\|\bz\|_2^2+\|\bq\|_2^2$ is a bad approximation to $\|\bx\|_2^2$, then because
\[\|\bx\|_2^2=\|\bz+\bq\|_2^2=\|\bz\|_2^2+2\langle\bz,\bq\rangle+\|\bq\|_2^2,\] $|\langle\bz,\bq\rangle|$ must be large in magnitude, relative to both $\|\bz\|_2^2$ and $\|\bq\|_2^2$.  
However, it can only happen when $\bz$ has a non-trivial alignment with $\bq$, in which case we can hope to learn something about $\bz$ from $\bq$. 
We can use $\bA\bx$ to check whether $\|\bz\|_2^2+\|\bq\|_2^2$ is a bad approximation, and either return an accurate estimate to $\|\bx\|_2^2$, or learn some extra information about $\bz$ in each step.
This intuition leads us to a version of our algorithm in a simplified setting.

\begin{figure*}[!thb]
\centering
\begin{tikzpicture}[scale=1.25]

\node at (4,2.5){Estimator};
\draw(3,2) rectangle+(2,1); 
\draw(4,3)[->] -- (4,3.5);
\node at (4,4.1){Estimate $\|\bz-\bz'\|_2^2+\|\bz'+\bq\|_2^2$};
\node at (4,3.7){as current estimate for $\|\bz+\bq\|_2^2$};

\draw(5,2)[dashed,->,color=purple] -- (6.8,1.2);
\node at (7.8,1.7){\textcolor{purple}{Estimate verification}};

\node at (7,0.5){Corrector};
\draw(6,0) rectangle+(2,1); 
\draw(8,0.5)[->] -- (8.5,0.5);
\node at (9.7,0.8){Use sketch $\bA(\bz+\bq)$};
\node at (9.7,0.5){to output if};
\node at (9.7,0.2){estimator incorrect};

\draw(6,0.5)[dashed,->,color=purple] -- (2.2,0.5);
\node at (4.1,0.8){\textcolor{purple}{Incorrect estimate signal}};

\node at (1,0.5){Learner};
\draw(0,0) rectangle+(2,1); 
\node at (-1,0.7){Update $\bz'$};
\node at (-1,0.3){on signal};

\draw(1,1)[dashed,->,color=purple] -- (2.8,1.8);
\node at (0.4,1.7){\textcolor{purple}{Pass current iterate $\bz'$}};

\node at (4,-0.5){\textcolor{purple}{At most $\O{\frac{1}{\eps^2}\log n}$ steps before convergence of $\bz'$ to $\bz$.}};
\end{tikzpicture}
\caption{Flowchart for learner, corrector, and estimator framework. 
The estimator outputs a current estimate, the corrector verifies its accuracy, and the learner updates its internal state $\bz'$ based on incorrect estimates.}
\label{fig:flowchart}
\end{figure*}

\paragraph{Estimator, corrector, learner framework.}
Consider the following setting: 
\begin{enumerate}
\item 
We are initially given a sketch $\bA\bz$, then vectors $\bq$ arrive in the stream, for which, we can track $\bA\bq$ and $\bB\bq$;
\item 
For each $\bq$, we must output an estimate of $\|\bz+\bq\|_2^2$;
\item 
We focus on only ``protecting'' $\bA$ from the adversary: for now, assume we have accurate estimates of the norms for any vectors sketched by $\bB$. 
In fact, it may be helpful to think of $\bB$ as the identity matrix (e.g. the algorithm knows each $\bq$ exactly), and the task is to produce an estimate of $\|\bz + \bq\|_2^2$ while using $\bA$ as few times as possible.
\end{enumerate}

Our simplified $F_2$ moment estimation algorithm for this setting consists of three main ingredients:
\begin{enumerate}
\item 
An estimator, whose role is to output estimates of $\|\bz+\bq\|_2^2$.
\item
A learner, whose role is to learn $\bz$ to facilitate the estimation. 
When the estimator is incorrect, the learner will gain information about $\bz$ from $\bq$. 
\item
A corrector, whose role is to inform the learner/estimator when the estimate is incorrect. 
\end{enumerate}
The corrector uses $\bA\bx$ to detect when an estimate is incorrect. 
The learner will maintain a vector $\bz'$, which is a linear combination of the queries $\bq$ on which the algorithm was incorrect, though algorithmically, only $\bA\bz'$ and $\bB\bz'$ are explicitly maintained. 
The estimator will output $\|\bz-\bz'\|_2^2+\|\bz'+\bq\|_2^2$ as an estimate for $\|\bx\|_2^2=\|\bz+\bq\|_2^2$, where the first term is estimated using $\bA\bz-\bA\bz'$, and the second term is estimated using $\bB\bz'+\bB\bq$.
As sanity checks, note that if $\bz'=\mathbf{0}^n$, i.e., we have no information about $\bz$, the estimator corresponds to the previous estimator $\|\bz\|_2^2+\|\bq\|_2^2$, and on the other hand if $\bz'=\bz$, i.e., we have successfully determined $\bz$ exactly, then the estimator corresponds to the desired quantity $\|\bz+\bq\|_2^2$. 
If the estimator was initially incorrect, we will simply return the estimate of the \textit{corrector} for this query based on $\bA\bx$.


\begin{figure*}[!thb]
\centering
\begin{tikzpicture}[scale=0.5]
    \node [minimum size=0, inner sep=0, label=below:{\scriptsize $\bz'$}] at (0,0) (zp) {};
    \node [minimum size=0, inner sep=0, label=below:{\scriptsize $\bz''$}] at (1,0) (zpp) {};
    \node [minimum size=0, inner sep=0, label=above:{\scriptsize $\bz$}] at (1, 4) (z) {};
    \node [minimum size=0, inner sep=0, label=below:{\scriptsize $-\bq$}] at (7, 0) (q) {};
    \draw (zp) [thick] -- (z);
    \draw (z) edge [thick,dotted] (zpp);
    \draw (zp) edge [thick,->] (0.8,0);
    \draw (z) edge [thick] (q); 
    \draw (q) edge [thick] (zp);
\end{tikzpicture}
\caption{If $\|\bz-\bz'\|_2^2+\|\bz'+\bq\|_2^2<(1-\eps)\|\bz+\bq\|_2^2$, moving $\bz'$ towards $-\bq$ reduces the distance to $\bz$.}
\label{fig:incorrect_estimate}
\end{figure*}

\paragraph{Robustness and progress.}
By standard minimax arguments, it suffices for us to consider a deterministic adversary. 
In the algorithm framework above, one vulnerability is that the corrector must use $\bA\bx$ to flag incorrect estimates. Each such interaction potentially leaks information about $\bA$, so after a large number of interactions, the corrector itself might be broken by adversarial inputs. 
To this end, suppose that the estimator is only incorrect at most $L$ times.
Then, we use the key insight from the bounded computation paths technique of \cite{Ben-EliezerJWY22}: note that there are at most $\binom{m}{L}$ possible sets of times where the corrector can flag incorrect estimates over a stream of length $m$. 
Thus, there are only at most $\binom{m}{L}$ adaptive input streams that a deterministic adversary can generate, based on the information that is output by the corrector.
Now, by setting the failure probability of the corrector to be $\delta = \frac{1}{\poly(n)}\cdot\binom{m}{L}^{-1}$, by a union bound, the corrector will be robust to all possible adversarial streams generated by the deterministic adversary. 
Since the space complexity of streaming algorithms often scales with $\log\frac{1}{\delta}$, the resulting space complexity will roughly have a linear dependence on $L$. 
Note that we also use the estimate of the \emph{corrector} $L$ times as the output and further use $\bA$ at most $L$ times to estimate $\|\bz-\bz'\|_2^2$ after each update to $\bz'$, but a sketch of size linear in $L$ can also be used $\O{L}$ times against an adaptive adversary (e.g., use a different block of rows for each query).
Thus, it remains to upper bound the total number $L$ of incorrect outputs made by the estimator. 

To do so, consider each time the estimate is inaccurate (see \Cref{fig:incorrect_estimate}).
In this case, $\bz-\bz'$ must be far from being orthogonal with $\bz'+\bq$. 
Thus, by moving $\bz'$ towards $-\bq$ (or away from $-\bq$), the distance to $\bz$ will decrease non-trivially, i.e., $\bz'$ approximates $\bz$ better.
Formally, we use $\|\bz-\bz'\|_2^2$ as the ``progress measure''.
Then, an incorrect estimate for $\|\bz + \bq \|_2^2$ implies 
\[\left\lvert\|\bz-\bz'\|_2^2+\|\bz'+\bq\|_2^2-\|\bz+\bq\|_2^2\right\rvert\ge\eps\cdot\|\bz+\bq\|_2^2.\]
By expanding the both sides, we must have \[| \langle \bz - \bz', \bz' + \bq \rangle | \geq \frac{\eps}{2} \| \bz + \bq\|_2^2 = \frac{\eps}{2} \left( \|\bz - \bz'\|_2^2 + \|\bq + \bz'\|_2^2 + 2\langle \bz - \bz', \bq + \bz' \rangle \right).\]
Suppose that $\langle \bz - \bz', \bz' + \bq \rangle \geq \frac{\eps}{2} \|\bz + \bq\|_2^2$, as the other case is very similar. In particular, this can be used to see that $\langle \bz - \bz', \bz' + \bq \rangle \geq \eps\cdot\|\bz - \bz'\|_2\cdot\|\bq + \bz'\|_2$. 
Now, when the algorithm was incorrect on query $\bq$, the vector $\bz'$ will be updated to $\bz'' = \bz' + \alpha(\bq + \bz')$ for some carefully chosen $\alpha$. We have 
\begin{align*}
\|\bz - \bz''\|_2^2 &= \|\bz - \bz' - \alpha(\bq + \bz')\|_2^2 \\ 
&= \|\bz - \bz'\|_2^2 + \alpha^2 \|\bq + \bz'\|_2^2 - 2\alpha \langle \bz - \bz', \bq + \bz'\rangle \\ 
& \leq \|\bz - \bz'\|_2^2 +\alpha^2 \|\bq + \bz'\|_2^2 - 2\alpha \cdot \eps \|\bz - \bz'\|_2 \cdot\|\bq + \bz'\|_2.
\end{align*}
\noindent
Then, by setting the step-size $\alpha = \frac{\eps \cdot \|\bz - \bz'\|_2}{\|\bq + \bz'\|_2}$, we get 
\[\|\bz - \bz''\|_2^2 \leq (1-\eps^2) \|\bz - \bz'\|_2^2.\]

Importantly, this step-size $\alpha$ can also be estimated up to a $(1\pm \eps)$ multiplicative factor  with high probability using robust $L_2$ sketches and a bounded computation paths argument. Since $\|\bz\|_2^2\le\poly(n)$ for a stream of length $m=\poly(n)$, it follows that $\bz'$ will be updated at most $\O{\frac{1}{\eps^2}\log n}$ times, which gives our upper bound on $L$.
Thus, the algorithm solves the problem in the simplified setting with only $\poly\left(\frac{1}{\eps},\log n\right)$ rows in $\bA$.

\paragraph{Recursive estimator.}
To extend the above approach to a full streaming algorithm, the main shortcoming is that while the algorithm with sketching matrix $\bB$ has access to $\bB\bq$ and can gradually learn $\bB\bz$, it must correctly answer a large number of queries based on $\bB\bq$ throughout the stream. 
For example, if $\bx=\bz+\bq$ is the frequency vector for the full stream of length $m$ and $\bz$ and $\bq$ each correspond to frequency vectors for substreams of length $\frac{m}{2}$, then $\bB$ must itself be robust to $\frac{m}{2}$ adaptive queries en route to $\bq$, and similarly $\bA$ must be robust to $\frac{m}{2}$ adaptive queries that ultimately form $\bz$. 
To address this, the natural approach is to recurse on both halves of the stream, corresponding to $\bz$ and $\bq$ respectively. 

Recall that our estimator is $\|\bz-\bz'\|_2^2+\|\bz'+\bq\|_2^2$, and we will maintain both $\bA\bz'$ and $\bB\bz'$ as we update $\bz'$.
The first term $\|\bz-\bz'\|_2^2$ is estimated from $\bA\bz-\bA\bz'$.
The key observation is that estimating the second term $\|\bz'+\bq\|_2^2$ is a task of the same form as the full problem, but on a shorter stream: given a (rounded) sketch $\bB\bz'$ for an unknown $\bz'$, estimate $\|\bz'+\bq\|_2^2$ for $\bq$ that arrives in a stream.
This allows us to recursively estimate the second term.

\paragraph{Two levels and beyond.} 
For example, consider a two-level scheme where we can write $\bq=\bq_1+\bq_2$, so that $\|\bz'+\bq\|_2^2=\|\bz'+\bq_1+\bq_2\|_2^2$. 
We can use a sketch matrix $\bB_1$ to track $\bB_1(\bz'+\bq_1)$ and a sketch matrix $\bB_2$ to track $\bB_2\bq_ 2$. 
Similarly, we implement a learner at this second level, whose role is to learn $\bz'+\bq_1$ over the course of the stream updates of $\bq_2$. 
In particular, the learner maintains a vector $\bu$, which is a linear combination of queries for which an estimator at the second level is inaccurate, when using $\|\bz'+\bq_1-\bu\|_2^2+\|\bu+\bq_2\|_2^2$ as an estimate for $\|\bz'+\bq_1+\bq_2\|_2^2$. By the above intuition, this estimate is incorrect if $\langle\bz'+\bq_1,\bq_2\rangle$ is large, meaning the algorithm for $\bq_2$ learns information about $\bz'+\bq_1$ from $\bq_2$. 

To check if the estimate is inaccurate, we can again interface the learner and estimator at the second level with a corrector for $\|\bz'+\bq_1+\bq_2\|_2^2$. 
One caveat is that the iterate $\bz'$ might change over the course of the stream for $\bq_2$ because the estimator at the first level may be incorrect for some queries $\bq$. 
To address this, we simply start another block at the second level whenever $\bz'$ is updated. 
That is, instead of having two blocks each corresponding to streams of length $\frac{m}{2}$, we dynamically form blocks when either the stream length has reached a multiple of $\frac{m}{2}$ or when $\bz'$ is updated to $\bz' + \alpha(\bq + \bz')$ for some incorrect query $\bq$. 
Because we previously upper bounded the total number of updates to iterate $\bz'$ by $L=\O{\frac{1}{\eps^2}\log n}$, then the number of blocks at the second level is at most $\O{\frac{1}{\eps^2}\log n}$.

We can now split the stream into $B+\O{\frac{1}{\eps^2}\log n}$ blocks of size at most $\frac{m}{B}$ for the two-level scheme, and more generally, we recursively split each block into $B$ blocks toward a full $B$-ary tree. To avoid the accumulation of errors over blocks, we can again query a corrector at each level that maintains $\bA(\bq_1+\bq_2+\ldots+\bq_k)$ after $k$ blocks. 
Then, the corrector at each level must interact with $L = \O{\frac{1}{\eps^2}\log n}$ adaptive queries for \textit{each} of the $B$ child blocks in the next level in response to incorrect estimates over the entire stream, as well as $B$ adaptive queries to correct the estimates at the end of each block. 
Since we require $B^H\ge m$ and $B \geq L$ by construction, and each corrector must handle $B \cdot L + B$ queries corresponding to incorrect estimates, it suffices to set $B=\O{\frac{1}{\eps^2}\log n}$ and $H=\log m=\O{\log n}$ for $m=\poly(n)$. 
As each sketch must now only be robust to $\O{B^2}$ queries, the corresponding space complexity is $\poly\left(\frac{1}{\eps},\log n\right)$ per sketch. 
Since at most $\poly(B,H)=\poly\left(\frac{1}{\eps},\log n\right)$ sketches are maintained at any given time, our algorithm has the desired space complexity. 

\begin{figure*}[!thb]
\centering
\begin{tikzpicture}[scale=1.25]
\draw(-1,4.2)[thick,->] -- (9,4.2);
\node at (4,4.5){Stream};

\draw(7,4.8)[color=blue] -- (7,1.2);
\draw(0,3.3) rectangle+(8,0.5); 
\fill[pattern=north east lines, draw=none] (7,3.3) rectangle (8,3.8);
\node at (4,3.55){$\bx=\bz+\bq$};
\draw(7,3.3)[thick,color=red] -- (7,3.8);
\node at (-0.75,3.55){Level $H$};

\draw(0,2.5) rectangle+(3.8,0.5); 
\node at (1.9,2.75){$\bz$};
\draw(4.2,2.5) rectangle+(3.8,0.5);
\fill[pattern=north east lines, draw=none] (7,2.5) rectangle (8,3);
\node at (5.6,2.75){$\bq$};
\draw(7,2.5)[thick,color=red] -- (7,3);
\node at (9.2,2.75){Guess $\bz'$ for $\bz$};
\draw plot [smooth] coordinates {(5.5,2.34) (7.8,2.3) (8.7,2.5)};
\draw[->, thick] (5.5,2.35) -- ++(0,0.15);
\node at (-0.95,2.75){Level $H-1$};

\draw(0,1.7) rectangle+(1.8,0.5); 
\draw(2,1.7) rectangle+(1.8,0.5); 
\draw(4.2,1.7) rectangle+(1.8,0.5); 
\node at (5.1,1.95){$\bq_0$};
\draw(6.2,1.7) rectangle+(1.8,0.5); 
\fill[pattern=north east lines, draw=none] (7,1.7) rectangle (8,2.2);
\node at (6.6,1.95){$\bq_1$};
\draw(7,1.7)[thick,color=red] -- (7,2.2);
\node at (9.5,1.95){Guess $\bz''$ for $\bz'+\bq_0$};
\draw plot [smooth] coordinates {(6.5,1.54) (7.8,1.5) (8.7,1.7)};
\draw[->, thick] (6.5,1.55) -- ++(0,0.15);
\node at (-0.95,1.95){Level $H-2$};

\node at (4,1.3){$\vdots$};
\end{tikzpicture}
\caption{Example of recursive tree structure on stream with $B=2$ blocks and $H$ levels. 
Level $H-1$ outputs $\|\bz-\bz'\|_2^2+\|\bz'+\bq\|_2^2$ as estimate for $\|\bx\|_2^2=\|\bz+\bq\|_2^2$. 
Level $H-2$ outputs $\|\bz'+\bq_0-\bz''\|_2^2+\|\bz''+\bq_1\|_2^2$ as estimate for $\|\bz'+\bq\|_2^2=\|\bz'+\bq_0+\bq_1\|_2^2$. }
\label{fig:tree}
\end{figure*}

\subsubsection{Approximate triangle inequality}
As a natural extension of our $F_2$ algorithm, we show an algorithm to robustly estimate functions $\calF$ that satisfy an approximate triangle inequality and have non-robust turnstile sketches.
Indeed, observe that for a function $\calF$ that satisfies triangle inequality and $\calF(\bv)=\calF(-\bv)$ for all $\bv\in\mathbb{R}^n$, if we seek an $\alpha$-approximation for a sufficiently large constant $\alpha>3$, then we will use $\calF(\bz-\bz')+\calF(\bz'+\bq)$ as the estimator for $\calF(\bz+\bq)$. 
By triangle inequality, this is always an overestimate. 
The estimator will only be incorrect when
\[\calF(\bz-\bz')+\calF(\bz'+\bq)\ge \alpha \cdot \calF(\bz+\bq).\]
In this case, by triangle inequality, $\calF(\bz+\bq)+\calF(\bz-\bz')\ge \calF(\bz'+\bq)$, and so $\calF(\bz+\bq)+2\calF(\bz-\bz')\ge\alpha\cdot \calF(\bz+\bq)$, we must have
\[\calF(\bz+\bq)\le\frac{2}{\alpha-1}\cdot \calF(\bz-\bz').\]
Now if we take $\bz''$ to be the closest vector to $\bz$ in the span of $\bq$ and $\bz'$ under the function $\cal{F}$, we have
\[\calF(\bz-\bz'')\le \calF(\bz+\bq)\le\frac{2}{\alpha-1}\cdot \calF(\bz-\bz'),\]
which shows progress. 
Thus, the learner can apply this iterative procedure of updating $\bz'$ will converge to $\bz$ within $\O{\log n}$ steps, after which the learner will have recovered $\bz$. 
A similar analysis can be applied for $\calF$ that only satisfies approximate triangle inequality.
We can then apply the same framework with a corrector, an estimator, and a learner and in fact, we can use recursion and design a robust algorithm that outputs a $\Theta(1)$-approximation to $\calF$ using $n^{1/C}$ levels for some constant $C>1$. 

\section{\texorpdfstring{$F_2$}{F2} Algorithm}
\label{sec:ftwo}
In this section, we describe and analyze our robust algorithm for $F_2$ moment estimation on adversarial insertion-deletion streams, culminating in the statement of \Cref{thm:ltwo}. 
We first recall the classical AMS algorithm for $F_2$ estimation:
\begin{theorem}[AMS Algorithm for $F_2$ estimation]
\cite{AlonMS99}
\label{thm:ams}
Let $\eps\in(0,1)$ be an accuracy parameter and $\delta\in(0,1)$ be a failure probability. 
Let $k=\O{\frac{1}{\eps^2}\log\frac{1}{\delta}}$. 
There exists a turnstile streaming algorithm $\AMS$ that maintains a linear sketch $Ax$, where $A \in \mathbb{R}^{k \times n}$ is drawn from an explicit distribution over random matrices, uses $\O{k\log n}$ bits of space, and with probability $1-\delta$, outputs a $(1+\eps)$-approximation to the $F_2$ moment on streams of length $m=\poly(n)$.
\end{theorem}

\begin{figure*}[!htb]
\begin{mdframed}
\textbf{Algorithm}:
\begin{enumerate}
\item 
Maintain a tree structure on the stream, so that the lowest level $1$ separates the stream into single updates, and each node has $B$ children. 
\item
Each new node in level $i$ is formed either when $B$ nodes in level $i-1$ have passed in the stream or when the iterate in the parent node (in level $i + 1$) has been updated
\item 
For each block $C_{i,j}$, use a matrix $\bB_{i,j}$ to sketch various vectors, e.g., the frequency vector in the block and the iterate vectors, as required by the subroutines
\item
For each update, let $H$ be the height of the tree and let $\bB_{H+1}$ be a sketch matrix at the top
\item
Output $\EstLevel(H+1,\bB_{H+1}\bz)$ where $\bz$ is the frequency vector for the stream
\end{enumerate}
\end{mdframed}
\caption{Adversarially robust $F_2$ norm estimation algorithm on insertion-deletion streams}
\label{fig:alg:ltwo}
\end{figure*}

\paragraph{Algorithm description.}  
The algorithm maintains a hierarchical tree over the stream, where the lowest level, i.e., level 1, corresponds to blocks of $B$ individual updates and higher levels group $B$ child nodes into blocks.  
Each block $C_{i,j}$ is sketched using a matrix $\bB_{i,j}$ to track the frequency vector and any intermediate iterates.  
As updates arrive, sketches are propagated up the tree, with the top-level sketch $\bB_{H+1}$ summarizing the entire stream.  
The $L_2$ estimate is then computed using a subroutine $\EstLevel$ at the top level of a recursive structure, combining the estimates from lower levels to produce a robust estimate of the norm. 
This tree structure limits the number of adaptive queries each sketch must handle.  
By choosing block size $B = \O{\frac{1}{\eps^2}\log n}$ and tree height $H = \O{\log n}$, the algorithm achieves $\poly(\frac{1}{\eps}, \log n)$ space while maintaining adversarial robustness.
The algorithm appears in \Cref{fig:alg:ltwo} and is illustrated by \Cref{fig:tree}. 

\paragraph{The $\EstLevel$ subroutine.}
We now describe the subroutine $\EstLevel$ and its guarantees. 
The $\EstLevel$ procedure is responsible for computing robust $L_2$ estimates at a given level $i$ of the tree.  
To avoid notational clutter, we consider the case where $\bq$ is the query vector at level $i$, but in general, the algorithm needs to further decompose $\bq=\bq_0+\bq_1$ into the component $\bq_0$ for the left siblings in the previous block and the component $\bq_1$ in the active block (the active blocks are the ancestors of the leaf corresponding to the current update in the stream). 
We also write $\bz$ to be the portion of the stream seen in previous blocks, so that the goal is to estimate $\|\bz+\bq\|_2^2$. 

\begin{algorithm}[!htb]
\caption{$\EstLevel(i,\bB_i\bz)$ for $F_2$ moment estimation}
\label{alg:est:level:ltwo}
\begin{algorithmic}[1]
\State{Let $\calP_i$ be the active blocks at level $i$}
\State{Let $\bB_i$ be the sketch matrix for $\calP_i$}
\State{}\Comment{$\left(1+\frac{\eps}{100H}\right)$-approximation for $L_2$ norm, robust to $\poly\left(\frac{1}{\eps},\log n\right)$ adaptive queries}
\State{Let $A_i$ be a $\left(1+\frac{\eps}{100H}\right)$-approximation to $\|\bz+\bq\|_2^2$}
\If{$i\neq 1$}
\State{$P_i\gets\|\bB_i\bz+\bB_i\bq_0-\bB_i\bv\|_2^2$}
\State{$Q_i\gets\EstLevel(i-1,\bB_{i-1}(\bv+\bq_1))$}
\Else
\State{$P_i\gets 0$}
\State{Let $\bq_1$ be the active part of the query in $\calP_i$}
\State{$Q_i\gets\|\bB_1\bv+\bB_1\bq_0+\bB_1\bq_1\|_2^2$}
\EndIf
\If{$P_i+Q_i\in\left(1\pm\frac{\eps}{100H}\right)^{3i}\cdot A_i$}
\State{\Return $P_i + Q_i$}
\Else
\State{\Return $A_i$}
\EndIf
\end{algorithmic}
\end{algorithm}


Intuitively, $\EstLevel$ maintains an estimator, a learner, and a corrector.  
The learner maintains a vector $\bv$, which is a linear combination of queries on which the estimator was previously incorrect, effectively learning about $\bz$ over time.  
The estimator attempts to output an approximation to $\|\bz+\bq\|_2^2$ using estimates to $\|\bz-\bv\|_2^2$ and $\|\bv+\bq\|_2^2$. 
The corrector monitors the sketches $\bB_i\bz$ and $\bA\bz$ to detect when the estimator is off, and informs the learner to adjust $\bv$. 
In particular, the sketch $\bB_i$ is initiated when the previous block at level $i+1$ is completed. 
Thus it tracks the frequency vectors induced by all left siblings of the active block at level $i$ and in particular $\bB_i\bz$. 

At a high level, $\EstLevel$ computes two contributions: $P_i$, corresponding to the current block corrected by the learner's iterate, and $Q_i$, recursively obtained from the lower level.  
Formally, $P_i$ estimates $\|\bz-\bv\|_2^2$ while $Q_i$ estimates $\|\bv+\bq\|_2^2$. 
If the sum $P_i + Q_i$ is consistent with an estimate $A_i$ of the desired quantity $\|\bz+\bq\|_2^2$, then the algorithm simply returns $P_i+Q_i$; otherwise, the algorithm outputs $A_i$. 
This mechanism ensures that each sketch only needs to handle a limited number of adaptive queries, while the recursive structure aggregates estimates from all levels to produce a robust and accurate $L_2$ norm for the entire stream.  
This concludes the discussion of the estimator and the corrector, whose interaction algorithm appears in \Cref{alg:est:level:ltwo}, where we use the notation $P_i+Q_i\in\left(1\pm\frac{\eps}{100H}\right)^{3i}\cdot A_i$ to denote the condition that $P_i+Q_i$ is contained within the interval $\left[\left(1-\frac{\eps}{100H}\right)^{3i}\cdot A_i,\left(1+\frac{\eps}{100H}\right)^{3i}\cdot A_i\right]$.

\paragraph{The $\MaintainIter$ subroutine.}
Finally, we describe the learner, which is represented by $\MaintainIter$, and is responsible for updating the iterate $\bv$ within a given active block, as an estimate for $\bz$. 
At the start of each block, $\bv$ is initialized to zero, representing no knowledge about the frequency vector $\bz$ from previous blocks.  
As updates arrive, the estimator computes a tentative estimate for the $L_2$ norm. 
If this estimate is inconsistent with the reference value $A_i$, the corrector signals that the estimator was inaccurate.  
$\MaintainIter$ then adjusts $\bv$ by adding a small, carefully scaled fraction of the current query vector, in the direction that would reduce the estimation error.  
This process allows the learner to gradually accumulate information about $\bz$ over the course of the block, while only working in the sketched spaces.  
The algorithm appears in full in \Cref{alg:mainiter:ltwo}, where we emphasize that all of these calculations are performed in the sketched space, i.e., under the image of a sketch matrix $\bB_i$ at level $i$, as the algorithm does not have explicit access to $\bv$ and $\bq$. 

\begin{algorithm}[!htb]
\caption{$\MaintainIter(i)$ for $F_2$ moment estimation}
\label{alg:mainiter:ltwo}
\begin{algorithmic}[1]
\State{Initialize $\bv=\textbf{0}^n$ at the beginning of the active block at level $i$}
\State{$\bq\gets\bq_0+\bq_1$}
\If{$P_i+Q_i\notin\left(1\pm\frac{\eps}{100H}\right)^{3i}\cdot A_i$}
\State{$\sigma \gets \textrm{sign}(A_i - (P_i + Q_i))$}
\State{$\alpha\gets\frac{1}{4}\sqrt{\frac{P_i}{Q_i}}\cdot\sigma\cdot\frac{\eps}{100H}$}
\State{$\bv\gets\bv+\alpha(\bv+\bq)$}
\EndIf
\State{\Return $\bv$ truncated to precision $\frac{1}{\poly(n)}$}
\Comment{All operations performed in the sketch space}
\end{algorithmic}
\end{algorithm}

The key step is showing that the learner can recover $\bz$ in a small number of updates, because each iteration corresponds to an incorrect output by the estimator, which therefore ``charges'' an additional adaptive interaction to the corrector. 
This is achieved by upper bounding the total number of iterations by the learner in \Cref{lem:bounded:iterations:ltwo}, as follows. 
Let $\bz'$ be the value of the iterate $\bv$ prior to an update. 
We observe that whenever the estimator outputs an inaccurate value for $\|\bz + \bq\|_2^2$, there must be correlation between the residual $\bz - \bz'$ and the current query $\bq + \bz'$, since an incorrect estimate implies 
\[|\langle \bz - \bz', \bq + \bz' \rangle| \ge \eps \cdot \|\bz - \bz'\|_2 \cdot \|\bq + \bz'\|_2.\]
Then by updating $\bz'$ to $\bz'' = \bz' + \alpha (\bq + \bz')$, where the step-size $\alpha$ is chosen proportional to the residual magnitude, possibly flipped in sign according to $\sigma = \textrm{sign}(\langle \bz - \bz', \bq + \bz'\rangle)$ and using $\|\bz - \bz'\|_2^2$ as a progress measure, one can verify that 
\[\|\bz - \bz''\|_2^2 \le (1-\eps^2) \|\bz - \bz'\|_2^2,\]
so each iteration makes nontrivial progress toward learning $\bz$. 
We remark the step-size $\alpha$ can be approximated within a $(1\pm \eps)$ factor using the estimate $P_i$ and $Q_i$ with high probability. 
Now since $\|\bz\|_2^2 \le \poly(n)$, the total number of iterations is upper bounded by $\O{\frac{1}{\eps^2} \log n}$. 

\paragraph{Algorithm analysis.}
We now justify the guarantees of our algorithm. 
To avoid notational clutter, we consider the case where $\calP_i$ is the left-most child of its parent, so that $\bq_0=\mathbf{0}^n$ and thus $\bz+\bq_0=\bz$, so it suffices for $\bv$ to approximate $\bz$. 
In general, we require $\bv$ to approximate $\bz+\bq_0$. 

We first show that $P_i$ in $\EstLevel$ is a good approximation to the quantity $\|\bz+\bq_0-\bv\|_2^2$. 
\begin{lemma}
\label{lem:est:proj}
Let $\bv$ be a fixed iterate vector, i.e., the value of $\bv$ in \Cref{alg:est:level:ltwo}, at a fixed time in the stream, conditioned on the previous times. Suppose that $P_i$ is  produced using sketch $\bB_i$ which is robust for $\O{\frac{1}{\eps^2} \log n}$ adaptive updates.
Then with high probability, 
\[\|\bz+\bq_0-\bv\|_2^2\le P_i\le\left(1+\frac{\eps}{100H}\right)\cdot\|\bz+\bq_0-\bv\|_2^2.\]
\end{lemma}
\begin{proof}
Recall that $\bB_i$ is a sketch matrix for $\left(1+\frac{\eps}{100H}\right)$-approximation. By the guarantees of the $\AMS$ algorithm in \Cref{thm:ams} and the robustness of $\AMS$ to $\O{\frac{1}{\eps^2} \log n}$ adaptive updates, by a bounded computation paths argument, c.f., \cite{Ben-EliezerJWY22}, we have
\[\|\bz+\bq_0-\bv\|_2^2\le P_i\le\left(1+\frac{\eps}{100H}\right)\cdot\|\bz+\bq_0-\bv\|_2^2.\]
\end{proof}

\FloatBarrier
Next, we show that $Q_i$ in $\EstLevel$ is a good approximation to $\|\bv+\bq_1\|_2^2$. 
Consequently, $P_i+Q_i$ is a good approximation to $\|\bz+\bq_0-\bv\|_2^2+\|\bv+\bq_1\|_2^2$. 
\begin{lemma}
\label{lem:estlvl}
For each $i\in[H]$, we have that with high probability, the output $P_i+Q_i$ of $\EstLevel(i,\bB_i\bz)$ satisfies
\[\|\bz+\bq_0-\bv\|_2^2+\|\bv+\bq_1\|_2^2\le P_i+Q_i\le\left(1+\frac{\eps}{100H}\right)^{3i-2}\cdot\left(\|\bz+\bq_0-\bv\|_2^2+\|\bv+\bq_1\|_2^2\right).\]
\end{lemma}
\begin{proof}
Suppose $\EstLevel(i-1,\bB_i(\bv+\bq_1))$ outputs $Q_i$ such that $\|\bv+\bq_1\|_2^2\le Q_i\le\left(1+\frac{\eps}{100H}\right)^{3(i-1)-2}\cdot\|\bv+\bq_1\|_2^2$. 
Then it follows by \Cref{lem:est:proj} that with high probability, the output $P_i+Q_i$ of $\EstLevel(i,\bB_{i+1}\bz)$ certainly satisfies
\[\|\bz+\bq_0-\bv\|_2^2+\|\bv+\bq_1\|_2^2\le P_i+Q_i\le\left(1+\frac{\eps}{100H}\right)^{3i-2}\cdot\left(\|\bz+\bq_0-\bv\|_2^2+\|\bv+\bq_1\|_2^2\right).\]
As the above argument serves as an inductive hypothesis for $i\in[H]$, it remains to analyze the base case of $i=1$. 
In this setting, we have $Q_1=\|\bB_1\bv+\bB_1\bq_0+\bB_1\bq_1\|_2^2$. 
By the correctness of $\bB_1$, we have 
\[\|\bv+\bq_0+\bq_1\|_2^2\le Q_1\le\left(1+\frac{\eps}{100H}\right)\cdot\|\bv+\bq_0+\bq_1\|_2^2,\]
with high probability. 
We have $3i-2\ge 1$ for $i=1$ and hence, the base case 
\[\|\bv+\bq_0+\bq_1\|_2^2\le Q_1\le\left(1+\frac{\eps}{100H}\right)^{3i-2}\cdot\|\bv+\bq_0+\bq_1\|_2^2,\]
is satisfied with high probability.
Then by a union bound, the desired claim holds for all $i\in[H]$ by induction. 
\end{proof}

\FloatBarrier

\noindent
Next, we show that conditioned on the correctness of the sketches at each node, we maintain an invariant that the estimator at each level $i$ is always correct, due to a corrector at level $i$ flagging inaccurate estimates. 
\begin{invariant}
\label{invar:output:acc:ptwo}
For any level $i\in[H]$, let $\calP$ be the active block in level $i$ and let $\bB_{i+1}$ be the sketch matrix of the parent of $\calP$. 
Let $\bq$ be the part of the query in $\calP$ and let $\bz$ be previous part. 
Then with high probability 
\[\|\bz+\bq\|_2^2\le\EstLevel(i,\bB_i\bz)\le\left(1+\frac{\eps}{100H}\right)^{3i+1}\cdot\|\bz+\bq\|_2^2.\]
\end{invariant}
\begin{proof}
Suppose the current iterate $\bv$ is not updated to $\bv + \alpha(\bq + \bv)$ for the current query $\bq$. 
Then by the construction of the algorithm, we must have
\[\calA_i(\bz+\bq)\le P_i+Q_i\le\left(1+\frac{\eps}{100H}\right)^{3i}\cdot\calA_i(\bz+\bq).\]
By the correctness of $\calA_i$, we have
\[\|\bz+\bq\|_2^2\le\calA_i(\bz+\bq)\le\left(1+\frac{\eps}{100H}\right)\cdot\|\bz+\bq\|_2^2.\]
Therefore,
\[\|\bz+\bq\|_2^2\le\EstLevel(i,\bB_i\bz)\le\left(1+\frac{\eps}{100H}\right)^{3i+1}\cdot\|\bz+\bq\|_2^2,\]
and so the invariant holds if $\bv$ is not updated at this step.

On the other hand, if the estimator was incorrect, we update iterate $\bv$ to $\bv' = \bv + \alpha(\bq + \bv)$, and return the estimate $A_i$, which is a $\left(1+\frac{\eps}{100H}\right)$-approximation to $\|\bz+\bq\|_2^2$. 
Hence, the invariant follows for all $i\in[H]$. 
\end{proof}
\FloatBarrier

\noindent
We now upper bound the number of times $L_i$ that the iterate $v$ in block $C_{i,j}$ at any level $i\in[H]$ is updated, across all times in the stream. 

\begin{lemma}[Bounded iterations]
\label{lem:bounded:iterations:ltwo}
Let $L_i$ be the number of times that the iterate $\bv$ in a particular block at level $i\in[H]$ is updated and let $\eta=\Theta\left(\frac{\eps}{H}\right)$. 
With high probability, $L_i\le\O{\frac{1}{\eta^2}\log n}$ for all $i\in[H]$ and for all times in the stream, for $m\le\poly(n)$. 
\end{lemma}
\begin{proof}
For any fixed level $i\in[H]$ and for any fixed time in the stream, let $\calP_i$ and $\calP_{i+1}$ be the active blocks at levels $i$ and $i+1$ respectively. 
Let $\bz$ be the frequency vector corresponding to all the updates prior to $\calP_{i+1}$. 
Conditioned on the correctness guarantees of $\EstLevel$ in \Cref{lem:estlvl}, we have that 
\begin{align*}
\|\bz - \bv\|_2^2 + \|\bv + \bq \|_2^2 \leq \EstLevel(i, \bB_{i+1}\bz) \leq (1+\eta) \cdot \left( \|\bz - \bv\|_2^2 + \|\bv + \bq \|_2^2\right).
\end{align*}
If the iterate $\bv$ is updated, then it must be that the estimate is inaccurate, so that
\[\left\lvert\|\bz-\bv\|_2^2+\|\bv+\bq\|_2^2-\|\bz+\bq\|_2^2\right\rvert\ge\eta\cdot\|\bz+\bq\|_2^2.\]
By expanding the left hand side, we get
\begin{align*}
|\langle \bz - \bv, \bv + \bq \rangle| &\geq \frac{\eta}{2} \|\bz + \bq\|_2^2\\ &=\frac{\eta}{2}\left(\|\bz - \bv\|_2^2 + 2 \langle \bz - \bv, \bv + \bq \rangle + \|\bv + \bq \|_2^2 \right).
\end{align*}
Suppose $\langle \bz -\bv, \bv + \bq \rangle \geq \frac{\eta}{2}\cdot \|\bz + \bq\|_2^2$. By rearranging the inequality above, we get 
\begin{align*}
(1-\eta) \cdot\langle \bz - \bv, \bv + \bq \rangle &\geq \frac{\eta}{2} \left(\|\bz - \bv\|_2^2 + \|\bv + \bq\|_2^2 \right) \\
& \geq \eta\cdot\|\bz - \bv\|_2\cdot \|\bv + \bq\|_2 
\end{align*}
Therefore, it follows that $\langle \bz - \bv, \bv + \bq \rangle \geq \eta\cdot\|\bz - \bv\|_2\cdot\|\bv + \bq\|_2$. 
Then, $\MaintainIter$ will update $\bv$ to $\bv' = \bv + \alpha(\bv + \bq)$ for some step-size parameter $\alpha \in [-1,1]$. 
To upper bound the total number of times this occurs, we will consider the progress measure $\|\bz - \bv \|_2^2 - \|\bz - \bv'\|_2^2$. 
First, we have that 
\begin{align*}
\|\bz - \bv'\|_2^2 = \|\bz - \bv - \alpha(\bv+\bq)\|_2^2 = \|\bz - \bv\|_2^2 - 2\alpha \langle \bz - \bv, \bv + \bq \rangle + \alpha^2\cdot\|\bv + \bq\|_2^2.
\end{align*}
So, it follows that 
\begin{align*}
\|\bz - \bv\|_2^2 - \|\bz - \bv'\|_2^2 &\geq 2\alpha \langle \bz - \bv, \bv + \bq \rangle - \alpha^2 \cdot\|\bv + \bq\|_2^2 \\ 
&\geq 2 \alpha \cdot \eta \|\bz - \bv\|_2 \cdot\|\bv + \bq \|_2 - \alpha^2\cdot \|\bv + \bq\|_2^2.
\end{align*}
For any $\alpha = \Theta(\eta) \cdot \frac{\|\bz - \bv\|_2}{\|\bv + \bq\|_2}$ where the constant ranges in $\left[\frac{1}{10},\frac{1}{2}\right)$, we get that 
\begin{align*}
\|\bz - \bv\|_2^2 - \|\bz - \bv'\|_2^2 \geq\frac{\eta^2}{100}\cdot\|\bz - \bv\|_2^2.
\end{align*}
Similarly, suppose that $\langle \bz - \bv, \bv + \bq \rangle \leq -\frac{\eta}{2}\cdot\|\bz + \bq\|_2^2$. 
Then, we see that 
\begin{align*}
\langle \bz - \bv, \bv + \bq \rangle &\leq -\frac{\eta}{2} \left(\|\bz - \bv\|_2^2 + 2 \langle \bz - \bv, \bv + \bq \rangle + \|\bv + \bq \|_2^2 \right).
\end{align*}
So, it follows that $(1+\eta) \cdot \langle \bz - \bv, \bv + \bq \rangle \leq  -\eta \cdot\|\bz - \bv\|_2\cdot\|\bq + \bv\|_2$, e.g. 
\[- \langle \bz - \bv, \bv + \bq \rangle \geq  \frac{\eta}{2}\cdot \|\bz - \bv\|_2\cdot\|\bq + \bv\|_2.\] 
As before, we consider the progress measure $\|\bz - \bv\|_2^2 - \|\bz - \bv'\|_2^2$: 
\begin{align*}
\|\bz - \bv\|_2^2 - \|\bz - \bv'\|_2^2 &\geq 2\alpha \langle \bz - \bv, \bv + \bq \rangle - \alpha^2\cdot \|\bq + \bv\|_2^2 .
\end{align*}
For $\alpha = -\Theta(\eta) \cdot \frac{\|\bz - \bv\|_2}{\|\bq + \bv\|_2}$ with any constant in the range $\left[\frac{1}{10},\frac{1}{2}\right)$, we have
\[\|\bz - \bv\|_2^2 - \|\bz - \bv' \|_2^2 \geq \frac{\eta^2}{100}\cdot\|\bz - \bv\|_2^2.\]
Thus, we have shown that $\|\bz - \bv'\|_2^2$ decreases by a factor of $(1-\O{\eta^2})$ in both cases. 
Observe that since $\|\bz\|_2^2\le\poly(n)$, then this holds even when the image of $\bv'$ is truncated to a sufficiently fine precision $\frac{1}{\poly(n)}$. 
Moreover, since $\bv = \mathbf{0}^n$ originally and $\|\bz \|_2^2 \leq \poly(n)$, we get that $\bv$ will be updated at most $\O{\frac{1}{\eta^2} \log n}$ times.
\end{proof}
\noindent
Finally, we now show correctness of \Cref{fig:alg:ltwo}, which gives our adversarially robust algorithm for $(1+\eps)$-approximation of the $F_2$ moment on insertion-deletion streams in polylogarithmic space. 
\thmltwo*
\begin{proof}
First, observe that a randomized adversary is a distribution over sequences of possible inputs and the probability that the attack succeeds is at most the maximum of the probabilities that any particular input sequence succeeds. 
Thus, without loss of generality, it suffices to consider a deterministic adversary. 

We first track the adaptive interactions with the linear sketches for the corrector. 
Note that at each time, $\MaintainIter$ asks each $F_2$ linear sketch $\bB_i$ whether to update the iterate $\bv$ with the query $\bq$. 
Therefore, the transcript from $\bB_i$ can be viewed as a sequence of symbols, either $\bot$ or $\top$, so that an output of $\bot$ from the sketching matrix indicates that no action should be taken, while an output of $\top$ indicates that the iterate should be updated. 
Suppose $\top$ occurs at most $L$ times in a single block; since there are $B$ blocks in the level below, the corrector will be used to flag at most $B \cdot L$ queries overall.

This latter action also triggers a linear sketch at level $i-1$ to maintain $\bB_{i-1}\bv'$, where $\bv'$ is the updated iterate $\bv'=\bv+\eps(\bq+\bv)$, releases the estimate produced by $\bB_i$, but no other information about $\bB_i$ is revealed. 
Moreover, the linear sketch outputs an estimate for $F_2$ moment at the beginning of each of the $B$ blocks, encoded in $C_1\log n$ bits, for some fixed constant $C_1>0$. 
We remark that the entries of the truncated image of $\bv$ must be multiples of the precision $\frac{1}{\poly(n)}$ with magnitude at most $\poly(n)$, since otherwise $\bv$ could not be a good approximation to $\bz$, which would violate the progress analysis in \Cref{lem:bounded:iterations:ltwo}. 
Hence, the total number of possible output streams induced by $\bB_{i}$ is at most $\binom{m}{B \cdot L}\cdot\left(2^{(C_1\log n)}\right)^{B \cdot L + B}$. 
Since the adversary is deterministic, then any possible input stream must be induced by one of these possible output streams. 
By setting the total failure probability $\delta$ so that
\[\delta\le\frac{1}{(nm)^3}\cdot\left(\binom{m}{B \cdot L}\cdot\left(2^{(C_1\log n)}\right)^{B \cdot L + B}\right)^{-1},\]
then by a union bound over all possible input streams from the adversary, the sketching matrix $\bB_i$ is correct with probability $1-\frac{1}{(nm)^3}$.  
For $m=\poly(n)$, it suffices to set $\log\frac{1}{\delta}=\O{(B \cdot L + B)\cdot\log n}$. 
This shows robustness, after which correctness holds from \Cref{invar:output:acc:ptwo}. 

We next track the number of adaptive interactions to the estimator, which consists of two components $P_i$ and $Q_i$. 
While the correctness of $Q_i$ follows from recursion, we again use a bounded computation paths argument to show correctness of $P_i$. 
To that end, we first remark that each sketch matrix $\bB_i$ has entries rounded to $\O{\log n}$ bits and thus each output $P_i$ can be represented in $\O{\log n}$ bits. 
The number of adaptive interactions with the sketch $\bB_i$ is at most the number of updates to the iterate $\bv$, which is at most $L$. 
Thus the number of possible computation paths is at most $\binom{m}{L}\cdot\left(2^{(C_2\cdot\log n)}\right)^L$ for some constant $C_2>0$ and so it again suffices to set $\log\frac{1}{\delta}=\O{(L+B)\cdot\log n}$. 
Specifically, each of the adaptive queries to $P_i$ correspond to different values of the iterate $\bv$; all other queries to $P_i$ over the course of the stream are to the same values of the iterate $\bv$ between two updates of $\bv$ and thus correct conditioned on the correctness of the queries for the values of $\bv$. 

\paragraph{Space complexity.}
It thus remains to analyze the space complexity. 
For a fixed tree and for each fixed level $i\in[H]$ within the tree, we maintain $B$ sketching matrices corresponding to the $B$ nodes for each active block in level $i+1$. 
Each sketch requires accuracy $\left(1+\eta\right)$ for $\eta = \frac{\eps}{100 H}$, and failure probability $\delta$, where from before, we have $\log\frac{1}{\delta}=\O{(L+B)\cdot B\cdot\log n}$. 
By \Cref{lem:bounded:iterations:ltwo}, we have $L\le\frac{C_3}{\eta^2}\log n$ for some fixed constant $C_3>0$ with high probability, conditioned on the correctness of the sketch. 
We require $B^H\ge m$ and $B\ge L$, so it suffices to set $B=\O{\frac{1}{\eta^2}\log n}$. 
By \Cref{thm:ams}, each linear sketch has sketching dimension $\tO{\frac{1}{\eta^2}\log\frac{1}{\delta}}$. 
Since 
\[\log\frac{1}{\delta}=\O{(B \cdot L + B)\cdot\log n}=\O{\frac{1}{\eta^4}\log^3 n},\]
this corresponds to sketching dimension $\tO{\frac{1}{\eta^6}\log^3 n}$ and therefore each sketching matrix uses $\tO{\frac{1}{\eta^6}\log^4 n}$ bits of space.  
There are $B=\O{\frac{1}{\eta^2}\log n}$ sketching matrices per level across $H=\O{\log n}$ levels in the recursion, and moreover, it suffices to maintain only one active sketch per level at any given time. 
Hence the total space is $\O{\frac{1}{\eta^6}\log^5 n}=\tO{\frac{1}{\eps^6}\log^{11} n}$ bits of space.  
\end{proof}

\begin{remark}
As an alternative to the computation paths argument used above, we can ensure that each sketch $\bB_i$ is robust to $B\cdot L + B= \O{B^2}$ adaptive interactions by applying the differential privacy-based framework in \cite{HassidimKMMS20} inside of each block $C_{i,j}$. 
By Theorem 3.4 in their work, this framework yields space 
\[\O{\frac{1}{\eta^2} \log\left(\frac{1}{\delta}\right) \log n \cdot \sqrt{(B \cdot L + B) \cdot \log \left(\frac{1}{\delta}\right)} \cdot \log\left(\frac{m}{\eta \delta}\right)},\]
for each sketch $\bB_i$. 
Since we maintain a single sketch in each of the $H = \O{\log n}$ levels at any given time, after substituting $B = \frac{1}{\eta^2} \log n$, $\eta = \frac{\eps}{H}$, and $\delta = \frac{1}{(nm)^3}$ for $m \leq \poly(n)$, the total space of our algorithm becomes:
\[\O{\frac{1}{\eps^4} \cdot \log^{9.5} n + \frac{1}{\eps^4} \log^{8.5} n \cdot \log\left(\frac{\log n}{\eps}\right)}\]
\end{remark}
\section{\texorpdfstring{$L_2$}{L2} Heavy-Hitters}
Next, we show a natural extension of our robust $F_2$ estimation algorithms to robustly compute the $L_2$ heavy hitters of the frequency vector at all times in the stream. 
Recall that an algorithm solves the $L_2$ heavy hitters problem if it returns coordinate $i$ whenever $|x_i|\geq \eps \|\bx\|_2$, and any coordinate $i$ with $|x_i|\leq \frac{\eps}{2} \|\bx\|_2$ is not returned. 
Note that as a direct corollary of our robust $L_2$ heavy hitters algorithm, we also obtain an adversarially robust algorithm that recovers the $L_p$ heavy hitters for all $p \leq 2$.

Suppose without loss of generality that $x_i>0$; the algorithm is symmetric for $x_i<0$. 
Intuitively, the main idea is to robustly approximate the $L_2$ norm $\|\bx \|_2\leq X \leq \left(1 + \frac{\eps}{100}\right)\|\bx\|_2$; then, to identify heavy hitters, we \textit{deterministically} iterate through each coordinate $i \in [n]$ of the frequency vector, add mass $\frac{1}{2} \eps \cdot X \cdot \be_i$ to coordinate $i$, and recompute the new $L_2$ norm $S_i$. 
If the approximate $F_2$ moment $S_i^2$ increases by at least $1.15 \cdot \eps^2 \cdot \|\bx\|_p^2$, we report $i$ as a heavy hitter, and move on to the next index. 
This algorithm is formalized in \Cref{fig:alg:heavy:hitter}.

\begin{figure*}[!htb]
\begin{mdframed}
\textbf{Algorithm}:
\begin{enumerate}
\item 
Let $\calA$ be a robust $L_2$ norm estimation algorithm with accuracy $(1+\O{\eps^2})$ 
\item 
Let $\bx^{(t)}$ denote the frequency vector for the stream at time $t$, and let $\|\bx^{(t)}\|_2 \leq X \leq \left(1+\frac{\eps^2}{100}\right) \|\bx^{(t)}\|_2$ be an estimate from $\calA$
\item
Initialize $H_t\gets\emptyset$
\item 
For $i \in [n]$:
\begin{enumerate}
\item 
$\bv_i\gets\frac{1}{2}\eps \cdot X \cdot \be_i$ for the elementary vector $\be_i$
\item 
Let $S_i \gets \calA(\bx^{(t)} + \bv_i)$ and $T_i\gets\calA(\bx^{(t)}-\bv_i)$
\item 
If $S_i^2 - X^2 \geq 1.15 \eps^2 \cdot X^2$ or $T_i^2 - X^2 \geq 1.15 \eps^2 \cdot X^2$, set $H \gets H \cup \{i\}$
\item 
Reset coordinate $i$, by inserting $-\bv_i$ or $\bv_i$ accordingly
\end{enumerate}
\item 
Return the set $H_t$ of heavy hitters at step $t$.
\end{enumerate}
\end{mdframed}
\caption{Algorithm for adversarially robust heavy-hitters on insertion-deletion streams.}
\label{fig:alg:heavy:hitter}
\end{figure*}

To justify the correctness of our algorithm, we first show that after adding roughly $\frac{\eps}{2}\cdot\|\bx\|_2$ to an $\eps$-heavy-hitter coordinate $i\in[n]$ will result in the overall $F_2$ moment increasing by at least a certain amount, while adding the same quantity to a coordinate $i\in[n]$ that is not $\frac{\eps}{2}$-heavy $i$ will result in the overall $F_2$ moment increasing by at most a smaller certain amount. 
These two amounts have a significant gap that a sufficiently accurate $F_2$ moment estimation can distinguish between the two cases, thereby certifying whether $i$ is at least $\eps$-heavy or at most $\frac{\eps}{2}$-heavy. 
\begin{lemma}
\label{lem:heavy:hitters}
Let $\bx\in\mathbb{R}^n$ be a frequency vector and suppose $\|\bx\|_2\le X\le\left(1+\frac{\eps^2}{100}\right)\cdot\|\bx\|_2$. 
Let $i\in[n]$ be fixed and let $Z$ be a $\left(1+\O{\eps^2}\right)$-approximation to $\left\|\bx+\frac{1}{2}\eps\cdot X\cdot \be_i\right\|_2^2$. 
Then with high probability:
\begin{enumerate}
\item 
If $x_i\ge\eps\cdot\|\bx\|_2$, then $Z-X^2>1.15\cdot\eps^2\cdot X^2$. 
\item
If $x_i\le\frac{\eps}{2}\cdot\|\bx\|_2$, then $Z-X^2<1.15\cdot\eps^2\cdot X^2$. 
\end{enumerate}
\end{lemma}
\begin{proof}
Let $\bx_{-i}$ be the vector $\bx$ with the $i$-entry set to zero, i.e., $\bx_{-i}:=\bx-x_i\cdot\be_i$ for the elementary vector $\be_i$. 
Then 
\[\|\bx\|_2^2=\|\bx_{-i}\|_2^2+(x_i)^2\]
and
\[\left\|\bx+\frac{1}{2}\eps\cdot X\cdot \be_i\right\|_2^2=\|\bx_{-i}\|_2^2+\left(\frac{1}{2}\eps\cdot X+x_i\right)^2.\]
Hence,
\begin{align*}
\left\|\bx+\frac{1}{2}\eps\cdot X\cdot\be_i\right\|_2^2-\|\bx\|_2^2&=\left(\frac{1}{2}\eps\cdot X+x_i\right)^2-(x_i)^2.
\end{align*}
For $x_i\le\frac{\eps}{2}\cdot\|\bx\|_2$, we have that
\[\left(\frac{1}{2}\eps\cdot X+x_i\right)^2-(x_i)^2\le(1.05^2-0.5^2)\cdot\eps^2\cdot\|\bx\|_2^2.\]
On the other hand, by convexity for $x_i\ge\eps\cdot\|\bx\|_2$, we have that 
\[\left(\frac{1}{2}\eps\cdot X+x_i\right)^2-(x_i)^2\ge(1.5^2-1)\cdot\eps^2\cdot\|\bx\|_2^2.\]
Note that there is a constant gap between these two cases, and thus these two cases can be distinguished by a $\left(1+\O{\eps^2}\right)$-approximation $Z$ to $\left\|\bx+\frac{1}{2}\eps\cdot X\cdot \be_i\right\|_2^2$. 
In particular, these two cases can be distinguished by the threshold $1.15\cdot Z$ for sufficiently small constant $\eps\in(0,1)$. 
\end{proof}
We now show the correctness of our adversarially robust heavy-hitter algorithm on insertion-deletion streams. 
\thmhh*
\begin{proof}
Since the algorithm maintains a robust $F_p$ moment estimation algorithm and makes deterministic queries to determine the set of heavy hitters at each time, robustness of this algorithm directly follows from \Cref{thm:ltwo}.
Additionally, by \Cref{lem:heavy:hitters}, a $\left(1 + \frac{\eps^2}{100}\right)$-approximation algorithm for $F_2$ moment estimation can distinguish between the case that $x_i \geq \eps \|\bx\|_2$ and $x_i \leq \frac{\eps}{2} \|\bx\|_2$. 
Thus, the total space bounds follow from the aforementioned theorem.
Finally, we remark that for $p\le 2$, the $\eps$-$L_p$ heavy hitters are a subset of the $\eps$-$L_2$-heavy hitters, as $|x_i|\ge\eps\cdot\|\bx\|_p$ implies $|x_i|^2\ge\eps^2\cdot\|\bx\|_p^2\ge\eps^2\cdot\|\bx\|_2^2$, since $p\le 2$. 
Thus, to identify the $\eps$-$L_p$ heavy hitters, we can simply run the $\eps$-$L_2$ heavy-hitter algorithm. 
\end{proof}

\section{Approximate Triangle Inequality}
In this section, we extend our robust framework for $F_2$ estimation to a broader class of functions $\calF$ that satisfy an approximate triangle inequality. 
Recall that for a constant $\beta>0$, a function $\calF$ satisfies a $\beta$-approximate triangle inequality if 
\[\calF(\bx-\bz)\le\beta\cdot(\calF(\bx-\by)+\calF(\by-\bz)),\]
for all vectors $\bx,\by,\bz\in\mathbb{R}^n$. 

As before, we use the framework consisting of an estimator, a corrector, and a learner, and the algorithm maintains an iterate $\bz'$ that approximates the previous stream vector $\bz$. 
Each incoming query $\bq$ is handled by estimating $P_i\approx\calF(\bz-\bz')$ and $Q_i\approx\calF(\bz'+\bq)$, the former using a non-adaptive sketch, and the latter using a recursive procedure.  
The estimator then outputs sum of these estimates $P_i+Q_i$, which by triangle inequality is at least $\calF(\bz+\bq)$. 
Now, if $P_i + Q_i$ is significantly larger than $\calF(\bz+\bq)$, then by the approximate triangle inequality, the iterate can be updated to a new vector $\bz''=-\bq$ in a way that provably reduces the distance to the true vector under $\calF$. 
Specifically, \Cref{lem:precond:tri} formalizes how such an update guarantees that $\calF(\bz-\bz'')$ is a constant factor smaller than $\calF(\bz-\bz')$, which measures progress toward learning $\bz$.  
As a result, \Cref{lem:bounded:iterations:tri} shows that the algorithm only performs $\O{\log n}$ updates per node $C_{i,j}$ level in level $i$ of the recursion. 
Finally, we again use a recursive structure to upper bound the number of adaptive interactions in each level. 
The important subroutine $\EstLevel$ appears in \Cref{alg:est:level:tri}; it is then encapsulated within the same tree structure as \Cref{fig:tree}. 

We begin by examining the precise condition under which the algorithm's output $\calF(\bz - \bz') + \calF(\bz' + \bq)$ is not an accurate estimate of $\calF(\bz + \bq)$. 
\begin{lemma}[Precondition triggering]
\label{lem:precond:tri}
Let $\calA$ be an algorithm that outputs a $\kappa$-approximation to a function $\calF$ that satisfies $\beta$-triangle inequality and let $Z=\calA(\bz-\bz')+\calA(\bz'+\bq)$. 
Suppose $Z>\kappa^2\cdot\calA(\bz+\bq)$. 
Then with high probability, for $\bz''=-\bq$, we have 
\[\calF(\bz-\bz'')\le\frac{\beta+1}{\kappa-\beta}\cdot\calF(\bz-\bz').\]
\end{lemma}
\begin{proof}
Conditioned on the correctness of the subroutines, we have with high probability:
\begin{enumerate}
\item
$Z\in\left[\calF(\bz-\bz')+\calF(\bz'+\bq),\kappa\cdot(\calF(\bz-\bz')+\calF(\bz'+\bq))\right]$
\item
$\calA(\bz+\bq)\in\left(\calF(\bz+\bq),\kappa\cdot\calF(\bz+\bq)\right]$.
\end{enumerate}
Thus if $Z>\kappa^2\cdot\calA(\bz+\bq)$, then 
\[\kappa\cdot(\calF(\bz-\bz')+\calF(\bz'+\bq))\ge Z>\kappa^2\cdot\calA(\bz+\bq)\ge\kappa^2\cdot\calF(\bz+\bq),\]
so that
\[\calF(\bz-\bz')+\calF(\bz'+\bq)\ge\kappa\cdot\calF(\bz+\bq).\]
By the $\beta$-approximate triangle inequality property of $\calF$, we have $\beta\cdot(\calF(\bz+\bq)+\calF(\bz'-\bz))\ge\calF(\bz'+\bq)$. 
Since $\calF$ is symmetric, then $\calF(\bz'-\bz)=\calF(\bz-\bz')$, so that
\[(\beta+1)\cdot\calF(\bz-\bz')\ge(\kappa-\beta)\cdot\calF(\bz+\bq).\]
Therefore, we have $\calF(\bz-(-\bq))\le\frac{\beta+1}{\kappa-\beta}\cdot\calF(\bz-\bz')$, which gives our notion of progress by updating the iterate from $\bz'$ to $\bz''=-\bq$. 
\end{proof}

Next, we show that when the estimator is inaccurate on a query, we have a notion of progress, in the sense that $\calF(\bz-\bz'')$ is a constant factor smaller than $\calF(\bz-\bz')$, where $\bz'$ and $\bz''$ are the iterates maintained before and after the query. Note that since $\mathcal{F}$ satisfies a $\beta$-approximate triangle inequality, we have $\mathcal{F}(\bz - \bz') + \mathcal{F}(\bz' + \bq) \ge \frac{1}{\beta} \cdot \mathcal{F}(\bz + \bq)$. As a result, we conclude that the number of iterations that the estimate $\bv$ to $\bz$ can change is at most $\O{\log n}$ times. 

\begin{lemma}[Bounded iterations]
\label{lem:bounded:iterations:tri}
Suppose $\calF$ is a function that, on a stream of length $m=\poly(n)$, has value $\left[\frac{1}{\poly(m)},\poly(m)\right]$. 
Let $L_i$ be the number of times that the iterate $\bz'$ is updated at level $i$. 
Then with high probability, $L_i\le\O{\log n}$ for all $i\in[H]$ and for all times in the stream. 
\end{lemma}
\begin{proof}
For any fixed level $i \in [H]$ and for any fixed time in the stream, let $\calP_i$ be the active blocks at level $i$. 
Let $\bz$ denote the frequency vector corresponding to all the updates prior to $\calP_i$ and let $\bz'$ be the previous iterate maintained by $\EstLevel$. 
Let $\bq$ be the active part of the query in $\calP_i$. 

Conditioned on the correctness of the subroutines in \Cref{alg:est:level:tri}, if the iterate $\bv$ is not updated, we have $P_i + Q_i \le \kappa^{3i} \cdot A_i$, where $P_i$ and $Q_i$ are estimates of $\calF(\bz-\bz')$ and $\calF(\bz'+\bq)$, respectively, and $A_i$ is a $\kappa^i$-approximation to $\calF(\bz + \bq)$.
If $\bv$ is updated, then the precondition of \Cref{lem:precond:tri} so that the algorithm replaces $\bz'$ by $\bz'' = -\bq$. 
By \Cref{lem:precond:tri}, we obtain a notion of progress:
\[\calF(\bz-\bz'')=\calF(\bz+\bq)\le\frac{\beta+1}{\kappa-\beta}\cdot\calF(\bz-\bz').\]
Thus, each update of $\bz'$ guarantees a multiplicative decrease in $\calF(\bz-\bz')$ by a factor of at least $1-\Omega(1)$, for $\kappa>2\beta+1$. 
Since initially $\calF(\bz) \le \poly(n)$ and the minimum possible value of $\calF$ is at least $\frac{1}{\poly(n)}$, it follows that the iterate $\bv$ can be updated at most $L_i\le\O{\log n}$ times at each level $i$. 
Applying a union bound over all $H$ levels gives the desired high-probability bound.
\end{proof}

\begin{algorithm}[H]
\caption{$\EstLevel(i,\bB_i\bz)$ for function $\calF$ with $\beta$-approximate triangle inequality}
\label{alg:est:level:tri}
\begin{algorithmic}[1]
\State{Let $\calP_i$ be the active blocks at level $i$}
\State{Let $\bq_1$ be the active query in $\calP_i$}
\State{Let $\bq_0$ be the part of the query in the left siblings of $\calP_i$}
\State{$\bq\gets\bq_0+\bq_1$}
\State{Let $\bB_i$ be the sketch matrix for $\calP_i$}
\State{}\Comment{$\kappa$-approximation for function $\calF$, robust to $\tO{n^{1/C}}$ adaptive queries}
\State{Let $A_i$ be a $\kappa$-approximation to $\calF(\bz+\bq)$}
\State{Let $\bv$ be the previous iterate}
\State{Let $P_i$ be an estimate of $\calF(\bz+\bq_0-\bv)$ using sketch matrix $\bB_i$}
\If{$i\neq 1$}
\State{$Q_i\gets\EstLevel(i-1,\bB_{i-1}(\bv+\bq_1))$}
\Else
\State{Let $Q_i$ be an estimate of $\calF(\bv+\bq_1)$ using $\bB_1$}
\EndIf
\If{$P_i+Q_i\le\kappa^{3i}\cdot A_i$}
\State{\Return $P_i + Q_i$}
\Else
\State{$\bv\gets-\bq_1$}
\Comment{Performed in the sketch space}
\State{\Return $A_i$}
\EndIf
\end{algorithmic}
\end{algorithm}

\noindent
We now justify the guarantees of our algorithm, which again uses $\EstLevel$ as a building block within the tree structure from \Cref{fig:tree}. 
\thmtri*
\begin{proof}
We first reduce to a deterministic adversary. 
A randomized adversary is simply a distribution over input streams, and the probability that it succeeds is at most the maximum over deterministic sequences. 
Hence, it suffices to consider a deterministic adversary.

We now track the interactions with the sketches maintained by \Cref{alg:est:level:tri}. 
At each time, $\EstLevel$ queries the sketch at level $i$ to determine whether to update the iterate $\bv$. 
The output of a sketch can be viewed as a symbol: $\bot$ if no update occurs, $\top$ if the iterate is updated. 
By \Cref{lem:bounded:iterations:tri}, the number of times $\top$ occurs at level $i$ is at most $L_i = \O{\log n}$ with high probability.  
Let $B$ denote the number of blocks in each level. Each update at level $i$ triggers a query to level $i-1$ but does not reveal additional information about level $i$. 
Moreover, at the beginning of each block, the sketch outputs a $\kappa^i$-approximation encoded in $\O{1}$ words. Since the corrector sketch at level $i$ will output an estimate at most $L_i$ times for each of the $B$ blocks in level $i+1$, the total number of possible output transcripts at level $i$ is at most $\binom{m}{B \cdot L_i} \cdot \O{1}^{B \cdot L_i +B}$.

Now, to achieve constant-factor approximation, we require $\kappa^{\O{H}}=\O{1}$, and thus we require $H=\O{1}$. 
In particular, for any constant $C>1$, there exists $H=\O{1}$ such that the sizes of all blocks at the bottom level are at most $n^{1/C}$. 
Moreover, we remark that each output by the estimator and by the estimate of $P_i$ is at most $\O{\log n}$ bits. 
Thus the number of possible output transcripts at the bottom level is at most $\left(2^{(C_1\log n)}\right)^{n^{1/C}+L}$ for a fixed constant $C_1>0$, while at the other levels, the number of possible transcripts is at most $\binom{m}{B \cdot L}\cdot\left(2^{(C_1\log n)}\right)^{B\cdot L + B}$, since there are $L$ adaptive interactions with the corrector for each of the $B$ blocks in the level below.
Thus if we set $B=\O{n^{1/C}}$, then it suffices to set the failure probability of each sketch so that $\log\frac{1}{\delta}\le\O{n^{1/C}}\cdot\log m$ for some constant $C>1$. 
Conditioned on the correctness of the sketches, the iterate updates always make progress by \Cref{lem:precond:tri}. 
Specifically, each update reduces the function $\calF(\bz-\bv)$ by a constant factor. 
Since $\calF$ takes values in $\left[\frac{1}{\poly(m)},\poly(m)\right]$, this guarantees that each level performs at most $\O{\log n}$ updates, which matches the bound $L_i$.

For space complexity, observe that each sketch requires failure probability $\log\frac{1}{\delta}\le\O{n^{1/C}}\cdot\log m$. 
Since each sketch uses space $S(n)\cdot\log\frac{1}{\delta}$, then each sketch uses space $\tO{n^{1/C}}\cdot S(n)$. 
In fact, despite setting $B=\O{n^{1/C}}$, the total space at a particular level is $\tO{n^{1/C}}\cdot S(n)$ because it suffices to maintain a single active sketch per level at any given time. 
Since there are $H=\O{1}$ levels, the resulting space bound follows. 

Finally, correctness follows from the triangle inequality argument: the iterate updates always ensure that the output is a $\kappa^{\O{C}}$-approximation to $\calF$ at all times. 
Hence, the algorithm is adversarially robust and achieves the claimed space and approximation guarantees.
\end{proof}

\paragraph{Applications.}
We now present several detailed examples of functions that satisfy the approximate triangle inequality. 
The first example is $\mathcal{F}(\bx) = \norm{\bx}_p^p$ for $p \ge 0$.
\begin{lemma}
\lemlab{lem:xp<1}
For any $\bx, \by \in \mathbb{R}^n$ and $p\in[0,1]$, we have $\norm{\bx + \by}_p^p \le \norm{\bx}_p^p + \norm{\by}_p^p$.
\end{lemma}
\begin{proof}
Recall that for $0\le p \le 1$, for any $x, y \in \mathbb{R}$, we have $|x + y|^p \le |x|^p + |y|^p$. 
Then, we observe that 
\[\norm{\bx + \by}_p^p =\sum_{i = 1}^n |x_i + y_i|^p \le \sum_{i = 1}^n (|x_i|^p + |y_i|^p )\le \norm{\bx}_p^p + \norm{\by}_p^p \;. \qedhere\]
\end{proof}

\begin{lemma}
\lemlab{lem:xp>1}
For any $\bx, \by \in \mathbb{R}^n$ and $p \ge 1$, we have $\norm{\bx + \by}_p^p \le 2^{p - 1}\left(\norm{\bx}_p^p + \norm{\by}_p^p\right)$.
\end{lemma}
\begin{proof}
Recall that for $p \ge 1$, for any $x, y \in \mathbb{R}$, we have $|x + y|^p \le 2^{p - 1} (|x|^p + |y|^p)$. Then, we observe that 
\[\norm{\bx + \by}_p^p =\sum_{i = 1}^n |x_i + y_i|^p \le \sum_{i = 1}^n 2^{p - 1} (|x_i|^p + |y_i|^p )\le 2^{p - 1} (\norm{\bx}_p^p + \norm{\by}_p^p) \;. \qedhere\]
\end{proof}

We next consider the function $\mathcal{F}$ that has the form $\mathcal{F}(\bx) = \sum_{i = 1}^n g(x_i)$ for some function $g$ such that $g(t) = f(t^2)$ for some Bernstein function $f$. 
Specifically, the following are several examples of such a function $g$, which are common in robust statistics.

\begin{itemize}
\item 
(Pseudo–Huber Loss): $g_{\tau}(x) =  \tau(\sqrt{1 + (x / \tau)^2} - 1)$ 
\item 
(Cauchy/Lorentzian Loss): $g_{\tau}(x) = \log (1 + x^2 / \tau)$
\item 
(Generalized Charbonnier): $g_{\tau} = (1 + x^2 / \tau)^\beta - 1$ for $0 < \beta \le 1$.
\item 
(Welsch/Leclerc Loss): $g_{\tau}(x) = 1 - e^{-x^2 / \tau}$ 
\item 
(Geman–McClure loss):  $ \displaystyle g_{\tau} =\frac{x^2}{x^2 + \tau}$
\end{itemize}

\begin{definition}[Bernstein function, e.g., \cite{schilling2012bernstein}]
\deflab{def:bernstein-function}
A function $f:(0,\infty)\to[0,\infty)$ is called a \emph{Bernstein function} if $f\in C^\infty(0,\infty)$ and its derivative $f'$ is completely monotone; that is, for all $n\in\mathbb{N}_0$ and all $x>0$,
\[(-1)^n f^{(n+1)}(x) \ge 0.\]
\end{definition}
Let $f$ be any Bernstein function with $f(0) = 0$. We now show that any function $g$ defined by $g(t) = f(t^2)$ satisfies the approximate triangle inequality with $\beta = 2$.
As in the proof of \lemref{lem:xp<1}, it suffices for us to consider the one-dimensional case.
\begin{lemma}
Suppose that the function $g(t) = f(t^2)$ for some Bernstein function $f$ and $g(0) = 0$. 
Then for any $a, b \in \mathbb{R}$, we have $g(a + b) \le 2((g(a) + g(b))$.
\end{lemma}
\begin{proof}
Since $f$ is a Bernstein function, we have $f'(x) \ge 0$ and $f''(x) \le 0$ for all $x>0$, so $f$ is increasing and concave. 
Also, $g(0)=0$ implies $f(0)=0$.  

Concavity and $f(0)=0$ imply that $f$ is subadditive on $[0,\infty)$, i.e., $f(x+y) \le f(x) + f(y)$ for all $x,y \ge 0$, and that $f(2x) \le 2 f(x)$ for all $x \ge 0$.  

Now, for any $a,b \in \mathbb{R}$, we have $(a+b)^2 \le 2 a^2 + 2 b^2$, so
\[g(a+b) = f((a+b)^2) \le f(2 a^2 + 2 b^2) \le f(2 a^2) + f(2 b^2) \le 2(f(a^2) + f(b^2)) = 2(g(a) + g(b)).\]
\end{proof}

Next, we recall several classical results from the streaming literature. 
Note that in order to apply our framework, we require a non-robust sketch for the chosen function $g$. 
The following lemma from \cite{BCWY16} describes a sufficient condition under which the function $g$ admits a sublinear-space (non-robust) streaming algorithm. 

\begin{lemma} [Zero-one law for normal functions, \cite{BCWY16}] 
\lemlab{lem:zero-one-law}
Let $g: \mathbb{Z}_{n \ge 0} \to \mathbb{R}$ be a slow-jumping, slow-dropping, and predictable function. 
There exists a one-pass turnstile streaming algorithm that provides a $(1 \pm \eps)$-approximation to the value of $\|\bx\|_g = \sum_{i = 1}^n g(x_i)$ in sub-polynomial space, where 
\begin{enumerate}
\item 
A function $g\in \mathcal{G}$ is \emph{slow-dropping} if and only if there exists a sub-polynomial function $h$ such that for $y \in \mathbb{N}$ and $x < y$, we have $g(x) \le g(y) h(y)$.
\item 
A function $g\in\mathcal{G}$ is slow-jumping if and only if there exists a sub-polynomial function $h$ such that for any $x<y$ we have $g(y)\le \lfloor y/x\rfloor^{2}\, h(\lfloor y/x\rfloor\, x)\, g(x)$. 
\item
$g \in \mathcal{G}$ is predictable if and only if for every sub-polynomial $\eps>0$ there exists a sub-polynomial function $h$ such that for all $x\in\mathbb{N}$ and $y\in[1,\,x/h(x))$, either $g|(x + y) - g(x)| \le \eps(x) g(x)$ or $g(y)\ge g(x)/h(x)$.
\end{enumerate}
In particular, if the function $h$ is polylogarithmic, the space used by the streaming algorithm will also be polylogarithmic.
\end{lemma}

In the following lemma, we show that for any Bernstein function $f$ and $g(t) = f(t^2)$ where $g(0) = 0$, $g$ satisfies the slow-dropping, slow-jumping, and predictable properties of \lemref{lem:zero-one-law} for some choice of the function $h$.

\begin{lemma}
\lemlab{lem:zero-one-law-h}
Suppose that the function $f:(0,\infty)\to[0,\infty)$ is a Bernstein function and $g(x) = f(x^2)$. 
Then
\begin{enumerate}
\item 
$g$ is slow-dropping with $h(x) = 1$.
\item 
$g$ is slow-jumping with $h(x) = 4$.
\item 
$g$ is predictable with $h(x) = \lceil 3/\eps(x) \rceil$.
\end{enumerate}
\end{lemma}
\begin{proof}
(1) Since $f$ is a Bernstein function, we have $f' \ge 0$, which implies that $f$ is non-decreasing. 
Therefore $g(x) = f(x^2)$ is also non-decreasing on positive input, and for $x<y$ we have $g(x) \le g(y)$, showing that $g$ is slow-dropping with $h(x) = 1$.

(2) Since $f$ is a Bernstein function, its derivative $f'$ exists and is non-increasing. 
Define
\[s(t) := \frac{f(t)-f(0)}{t} = \frac{1}{t} \int_0^t f'(u)\,du, \qquad t>0.\]
Because $f'$ is non-increasing, $s(t)$ is non-increasing, so for $0<x<y$ we have $s(y) \le s(x)$. 
Therefore,
\[\frac{f(y)}{y} = \frac{f(0)}{y} + s(y) \le \frac{f(0)}{y} + s(x) \le \frac{f(0)}{x} + s(x) = \frac{f(x)}{x}.\]
This shows that $f(y) \le \frac{y}{x} f(x)$ for $y\ge x>0$, and hence
\[g(y) = f(y^2) \le \frac{y^2}{x^2} f(x^2) = \frac{y^2}{x^2} g(x).\]
Thus for the relevant input range, we can upper bound $(y/x)^2 \le 4$, showing that $g$ slow-jumping with $h(x) = 4$.

(3) Since $f$ is concave and increasing with $f(0)\ge 0$, we have
\[f'(t) \le \frac{f(t)}{t} \qquad (t>0).\]
Let $g(x) = f(x^2)$ and write $\Delta := (x+y)^2 - x^2 = 2xy + y^2$. 
Then
\[\frac{|g(x+y)-g(x)|}{g(x)} = \frac{f(x^2+\Delta)-f(x^2)}{f(x^2)}\le \frac{f'(x^2)}{f(x^2)} \Delta\le \frac{\Delta}{x^2} = 2\frac{y}{x} + \left(\frac{y}{x}\right)^2.\]
With $y \le x/h(x) \le (\eps(x)/3)\,x$, we have $y/x \le \eps(x)/3$, hence
\[2\frac{y}{x} + \left(\frac{y}{x}\right)^2 \le \frac{2}{3}\eps(x) + \frac{1}{9}\eps(x)^2 \le \frac{7}{9}\eps(x) < \eps(x),\]
where we used $\eps(x) \le 1$. Therefore
\[\frac{|g(x+y)-g(x)|}{g(x)} \le \eps(x),\]
which shows that $g$ is predictable with $h(x) = \lceil 3/\eps(x) \rceil$.
\end{proof}
Finally, we combine \lemref{lem:zero-one-law} and~\lemref{lem:zero-one-law-h} to conclude that there exists a non-robust turnstile streaming algorithm for functions $g$ satisfying \lemref{lem:zero-one-law-h}.

\begin{lemma}
Given a function $g: \mathbb{Z}_{n \ge 0} \to \mathbb{R}$, where $g(x) = f(x^2 )$ for some Bernstein function $f$, there exists a one-pass (non-robust) turnstile streaming algorithm that with high probability, outputs a $(1 \pm \eps)$-approximation to $\|\bx\|_g = \sum_{i = 1}^n g(x_i)$ using space $\poly\left(\frac{1}{\eps}, \log n\right)$.
\end{lemma}

\section*{Acknowledgements}
Elena Gribelyuk and Huacheng Yu are supported in part by an NSF CAREER award CCF-2339942. 
Honghao Lin was supported in part by a Simons Investigator Award and a CMU Paul and James Wang Sercomm Presidential Graduate Fellowship. 
David P. Woodruff is supported in part by Office of Naval Research award number N000142112647 and a Simons Investigator Award. 
Samson Zhou is supported in part by NSF CCF-2335411 and gratefully acknowledges funding provided by the Oak Ridge Associated Universities (ORAU) Ralph E. Powe Junior Faculty Enhancement Award.

\def\shortbib{0}
\bibliographystyle{alpha}
\bibliography{references}
\end{document}